\documentclass[11pt]{article}
\usepackage{setspace}
\usepackage[left=2.5cm,right=2.5cm,top=2.5cm,bottom=2.5cm]{geometry}
\usepackage{amssymb}
\usepackage{amsthm}
\usepackage{amsmath}
\usepackage{graphicx}
\usepackage{color}
\usepackage{relsize}
\usepackage{booktabs}
\usepackage{caption}
\usepackage{natbib}
\usepackage{apalike}
\newlength{\figurewidth}
\newlength{\figureheight}
\newtheorem{proposition}{Proposition}[section]

\newtheorem{theorem}{Theorem}[section]
\newtheorem{lemma}{Lemma}[section]
\newtheorem{corollary}{Corollary}[section]
\newtheorem{remark}{Remark}[section]
\newtheorem{assumption}{Assumption}
\begin{document} 
	\setcounter{page}{1}
	
	\title{Optimal auction duration:\\
	A price formation viewpoint}

	\author{
		Paul Jusselin\footnote{paul.jusselin@polytechnique.edu},~ Thibaut Mastrolia\footnote{thibaut.mastrolia@polytechnique.edu}~~and Mathieu Rosenbaum\footnote{mathieu.rosenbaum@polytechnique.edu} \\
		\'Ecole Polytechnique, CMAP
		}

	\maketitle
	\begin{center}
	\textbf{\vspace{2em}Version with supplementary material}
	\end{center}
		\begin{abstract}
\noindent	We consider an auction market in which market makers fill the order book during a given time period while some other investors send market orders. We define the clearing price of the auction as the price maximizing the exchanged volume at the clearing time according to the supply and demand of each market participants. Then we derive in a semi-explicit form the error made between this clearing price and the efficient price as a function of the auction duration. We study the impact of the behavior of market takers on this error. To do so we consider the case of naive market takers and that of rational market takers playing a Nash equilibrium to minimize their transaction costs.  We compute the optimal duration of the auctions for 77 stocks traded on Euronext and compare the quality of price formation process under this optimal value to the case of a continuous limit order book. Continuous limit order books are found to be usually sub-optimal. However, in term of our metric, they only moderately impair the quality of price formation process. Order of magnitude of optimal auction durations is from 2 to 10 minutes.
\end{abstract}


\section{Introduction}
	
In most historical (lit) markets, trading operates through a continuous-time	double auction system: the continuous limit order book (CLOB). This mechanism allows market participants to buy or sell shares at any time point at the quoted prices. However market orders systematically pay (at least) the spread as transaction cost. Moreover volumes impact prices as market makers readjust their positions in reaction to the order flow, resulting in additional trading costs. Beyond this, it has been argued that some mechanical flaws are inherent to CLOBs, particularly in the context of high frequency trading. The debate started in the academic literature notably with the very influential paper \cite{budish2015high}, see also \cite{farmer2012review,wah2013latency}. In this work, the authors explain that CLOBs lead to obvious mechanical arbitrage and generate a competition in speed rather than in price between high frequency market makers, to the detriment of final investors. They convincingly show that frequent batch auctions could be a way to remedy these flaws.\\

\noindent The idea that auctions could be a suitable mechanism for the functioning of financial markets is not new. For example, in the important paper \cite{madhavan1992trading}, see also \cite{garbade1979structural}, the interest of auctions compared to CLOBs is already investigated. Of course the discussion in this work is not about high frequency arbitrage opportunities, but rather on the fact that auctions could be beneficial for the price formation process by enabling investors to trade directly between each others, avoiding to pay spread costs to market makers.\\

\noindent In \cite{budish2015high}, the authors provide the order of magnitude of a lower bound for auction period leading to elimination of the high frequency flaws of CLOBs (about 100 milliseconds). However, the mentioned earlier literature suggests that longer auction times could be suitable, but usually without giving figures. This is why, quoting \cite{budish2015high}, \textit{ developing a richer understanding of the costs of lengthening the time between auctions is an important topic}. This is exactly what we do in this paper. We provide a sound and operational quantitative analysis of the optimal auction duration on a financial market, and compare the efficiency of this mechanism with that of a CLOB. We work with a criterion based on quality of the price formation process as in \cite{madhavan1992trading}, but in the context of modern high frequency markets as in \cite{budish2015high}. Thus we hope to bridge the gap between these two seminal papers.\\

\noindent Actually, there seems to be a growing interest in practice for trading outside standard CLOBs. For example, a very important fraction of trading activity is still made over the counter and a rising part of market participants turns to new forms of market structures such as dark pools or auctions. Some auctions are already organized regularly in many markets where the main mechanism is a CLOB, typically at the beginning and at the end of the trading day. There also exist auction markets where auctions take place one after the other all along the day, and without continuous trading phase between two auctions. During an auction, market participants can send and cancel limit or market orders. Then at a certain time (possibly random), a clearing price is fixed in order to maximize the exchanged volume and matched orders are executed at this price. This is for example the case in the BATS-Cboe periodic auctions market for European equities. In this market, auctions are triggered when a first order is sent (limit or market). Then settlement takes place at a random time such that the auction cannot last more than a pre-fixed duration (100 milliseconds)\footnote{https://markets.cboe.com/europe/equities/trading/periodic\_auctions\_book}.\\
	
\noindent In an auction context, the key issue for a regulator or an exchange is to set a proper time period for the auction, and to compare the relevance of this mechanism with that of a CLOB. In \cite{du2017optimal}, the authors study the efficiency of an auction market with respect to the duration of the auction. They propose a microscopic agent-based model with deterministic or stochastic arrival of private and public information. Agents optimize their demand schedules with respect to their information and some personal characteristics. The average utility over all agents is used as a criterion to prove that the optimal auction duration is related to the law of exogenous information arrival. The authors also study the impact of heterogenous speeds of agents. They show that fast agents prefer short auction durations while slow ones prefer long ones. However, in the case of heterogenous agents, they do not give any results on what the optimal auction duration should be.\\

\noindent Most other works on this topic use a price formation point of view to assess the quality of the specification of an auction. This is the case in \cite{garbade1979structural} where the authors propose a simple model for price formation in an auction market. The average squared difference between an efficient price and the clearing price is used as a metric to show that a positive optimal auction duration always exists. The suggested optimal duration is a trade-off between averaging effect (a long duration allows a large number of agents to take part in the auction, hence reducing uncertainty about the efficient price) and volatility risk (a short duration leads to small volatility risk). This model has been refined in \cite{fricke2018too}. In this paper, the authors investigate several generalizations of this framework such as the multi-assets case or the presence of a market maker using filtering techniques and observing correlated assets to infer the efficient price at the clearing time. \\

In our work, the same driving forces as in \cite{fricke2018too,garbade1979structural} will be key to define our optimal durations: averaging effect \textit{versus} volatility risk. However, an important limitation in \cite{fricke2018too,garbade1979structural} is that no market orders are considered so that all the agents can be seen as liquidity providers. It is necessary to relax this assumption since one observes a large part of market orders in the trading flows of actual auctions, see \cite{boussetta2017role}. For example, market participants having a marked to market benchmark or seeking for priority in execution may typically use market orders. This type of orders will have a crucial role when computing optimal auction durations. This is because long durations induce large variance in the imbalance of the market order flow leading to less accurate fixing prices. \\

\noindent Another important remark is that in \cite{du2017optimal,fricke2018too,garbade1979structural}, no comparison between the auction and CLOB markets can be made. This is because the CLOB structure is not included in the range of their models. They obtain optimal durations for auction markets, but cannot say wether CLOB markets are sub-optimal or not. In our modeling, CLOBs exactly correspond to auctions with duration equal to zero, making the comparison between auctions and CLOBs possible.\\
	
\noindent In this paper, inspired by the cited earlier literature, we take price discovery as our criterion to compute an optimal auction duration. Our approach extends in several directions those in \cite{fricke2018too,garbade1979structural,madhavan1992trading} and goes as follows. We consider a regenerative auction market with auctions starting when a market order is sent and with constant duration $h$. More precisely, we assume that after the $(i-1)-$th auction clearing (ended at time $\tau^{cl}_{i-1}$) the limit order book is emptied and a new market phase starts independently of the past. A new auction opens at time $\tau^{op}_{i}$ when a first market order is sent. This new auction ends at time $\tau^{cl}_{i } = \tau^{op}_{i} + h$. Our model encompasses both CLOB and auction market structures since CLOB corresponds to an auction with duration $0$ (because auctions are triggered by the arrival of a market order, as in several actual markets, and we assume that the LOB is never empty at the auction clearing).\\

\noindent In CLOB markets, there is competition between market makers optimizing their quotes and market takers search for suitable execution times. In auction markets, market takers have an additional possibility to access cheap liquidity: they can try to match their orders with other market orders sent in the opposite direction. For example if a large volume of buy market orders is sent before the auction clearing, it is a good opportunity for selling market takers to execute their orders at a good price. In this context, a new form of competition between buying and selling market takers may arise, with market makers playing a side role. We also investigate this situation where market takers are strategic and act optimally in order to reduce their trading costs. We notably show that there exists a Nash equilibrium for this game. In this framework, we can compute the function $E$ too, and thus find an optimal auction duration.\\

\noindent From a mathematical point of view, the existence of a Nash equilibrium boils down to solving a fully coupled multi-dimensional BSDE driven by counting processes with non-positive discontinuous generator. The existence of Nash equilibria associated to a system of BSDEs with discontinuous generator has been notably studied in \cite{hamadene2014bang} in a Brownian framework. BSDEs related to those in \cite{hamadene2014bang} have been essentially investigated in the one-dimensional Brownian case (see for instance \cite{jia2008class,fan2012one,duan2013bsdes,tian2013lp}) considering a semi-continuous generator with respect to the $Y$ process and Lipschitz continuous with respect to $Z$. Existence of solutions to these BSDEs can usually be obtained which is not the case for uniqueness, see for instance \cite[Remark 4]{jia2008class}. An extension to the multidimensional case (still considering Lipschitz continuous generators with respect to $Z$) is proposed in \cite{fan2018existence} together with a uniqueness result. In \cite{heyne2014minimal} the authors succeed in weakening the continuity assumption with respect to the $Z$ component and prove the existence and uniqueness of a minimal solution in the one-dimensional case under positive generator or a relaxed condition, see \cite[Section 3.3]{heyne2014minimal}. Up to our knowledge our existence result is new and extends \cite{hamadene2014bang} to the case of counting processes. As explained above the question of uniqueness is particularly intricate and out of the scope of this paper.\\
	
\noindent Finally, based on Euronext exchange data, we use our model to compute the optimal auction duration according to our price discovery criterion for 77 European stocks traded on Euronext. The first striking result is that the suggested durations are much larger than a few milliseconds, rather of order of $1$ to $5$ minutes. The second one is that in term of our metric, CLOB are indeed sub-optimal. However, the quality of the price formation process in CLOB market is not very far from that of the auction with optimal duration. Of course this work is only a first step towards a full analysis of the auction issue since we focus here on one specific (but crucial) criterion. Other aspects such as liquidity, tick size effects and fees or potential arbitrage opportunities should certainly be addressed in future works. We also neglect potential optimization of market makers strategies who could for example revise their quotes during the auction according to the current market orders imbalance. Nevertheless, we believe our results are original and striking enough to help exchanges and financial authorities rethink their policies in terms of market design.\\
	
The paper is organized as follows. In Section \ref{sec:model} we describe the auction mechanism and our model. We also provide our first main result on the computation of the function $E$. The case of strategic market takers optimizing their trading cost is considered in Section \ref{sec:imbalance_process}. Our calibration methodology and numerical results on equity data can be found in Section \ref{sec:results}. Section 5 provides financial insights. Proofs are relegated to an Appendix.

\section{The model}
\label{sec:model}

In this section, we introduce our model for auction market. We describe the organization of the market and the behavior of the two types of agents: market makers filling the limit order book (LOB) with limit orders and market takers sending market orders. Then we explain the clearing rule and compute the clearing price. Finally we provide a semi-explicit expression for the quality of the price formation process.

\subsection{Auction market design}
\label{subsec:organization_market}

We consider an auction market organized in independent sequential auctions triggered by market orders. More precisely, after the opening of the market or after the clearing of an auction, a new auction starts when a first market order is sent. We write $(\tau^{op}_i)_{i\in \mathbb{N}^*}$ for the sequence of opening times of the auctions and $(\tau^{cl}_i)_{i\in \mathbb{N}}$, with $\tau^{cl}_0=0$, for the sequence of clearing times. An auction has a duration of $h$ seconds and allows market takers to meet. When $h=0$, our model corresponds to a CLOB market since any market order is matched against the limit orders present in the LOB.\\

\noindent For a given auction starting at some time $\tau_i^{op}$, market participants can send market or limit orders. At the auction clearing time $\tau_i^{cl}=\tau_{i}^{op} + h$, a clearing price, denoted by $P^{cl}_{\tau_i^{cl}}$, is set to maximize the exchanged volume. More precisely, sellers who are willing to sell below the price $P^{cl}_{\tau_i^{cl}}$ sell their shares to buyers who are willing to buy above $P^{cl}_{\tau_i^{cl}}$. Each cleared share is sold at the clearing price, independently of the posted price of the associated limit order if any (to the benefit of participants sending limit orders).

\subsection{Market makers and market takers}
\label{subsec:market_makers_takers}
Along the day, market makers arrive randomly in the market and send limit orders to fill the LOB. During the $i-$th market phase market makers arrival times are given by $(\tau_{i-1}^{cl}+ \tau^{i, mm}_k)_{k \geq 0}$ where $\tau^{i,mm}_k$ is the $k-$th event time of a counting process $(N^{i, mm}_s)_{s\geq 0 }$. We describe the liquidity provided by the $k-$th market maker by its supply function $S_k$.  This supply function depends on the market maker's view on the efficient price  at the clearing time when he sends his limit orders. The efficient price process is $(P_{s})_{s\geq 0 }$ and can be seen as the average of market makers' opinions at a given time on the ``fair" value of the underlying asset. It satisfies $P_s=P_0+\sigma_f W_s$ with $W$ a Brownian motion and $\sigma_f$ a positive constant. When positive, the quantity $S_k(p)$ represents the number of shares the $k-$th market maker is willing to sell at price $p$ or above. Negative values correspond to shares he is willing to buy at price $p$ or below. We assume that
	$$
	S_k(p) = K(p-\tilde{P}_{k}), \text{ with } \tilde{P}_{k} =  \mathbb{E}[P_{\tau^{cl}_i}|\mathcal{F}_{\tau^{cl}_{i-1}+\tau^{i,mm}_{k}}]   + g_k,
	$$ 
where $\tilde{P}_k$ is the view on the price of the asset by the $k-$th market maker when he sends his orders and $K$ a positive constant, $(g_k)_{k> 0}$ a sequence of i.i.d random variables with variance $\sigma^2$ representing the estimation noise in the inference of the efficient price by the market maker independent of all other processes. Linear supply functions are also considered in \cite{du2017optimal,fricke2018too}. This is equivalent to assume that each market maker sends uniform limit sell order above price $\tilde{P}_k$  and uniform limit buy orders to price level below $\tilde{P}_k$. In this case we get
$$
\tilde{P}_k  = P_{\tau_{i-1}^{cl} + \tau^{i, mm}_k}  + g_k. 
$$

\noindent  In practice there are different kinds of market makers and we could have assumed that each market maker has its own noise. That said, there are typically not so many market makers in the market and since they basically have the same technology, it is reasonable to assume that they have the same noise parameter. Note also that a model with different variance parameters would be very hard to calibrate because of the anonymity of our data.
\\

\noindent Consequently, market makers inject information in the LOB since they reveal the knowledge they have on the price through their supply function. However, the longer the auction duration, the less reliable the views of market makers arrived early in term of the estimation of the efficient price $P$ at the clearing time\footnote{To partially address this issue we extend our model allowing market makers to revise their position by canceling their orders in Appendix \ref{appendix:cancellation}.}\footnote{Note that a possible extension would be to consider that $\tilde{P}_k$ also depends on recently observed clearing prices, see \cite{fricke2018too}.}. Finally to obtain a regenerative market we consider that after the auction clearing time $\tau^{cl}_i$ market makers cancel their unmatched limit orders. Since sequential auctions markets with sufficiently large durations do not really exist, it is hard to have an idea of what would be the market maker's behavior. Of course, total cancellation after the clearing is not so realistic. However, note that the times of interest of our analysis are the auctions closing times, where the model is very reasonable. For example, when $h=0$, which corresponds to a CLOB market, at each time a market order is sent, the LOB is already filled thanks to Assumption 2 below. By setting $\Delta_i = \tau^{cl}_i-\tau^{cl}_{i-1}$ we deduce that at the clearing there is $N^{i, mm}_{\Delta_i}$ market makers in the LOB.\\

\noindent During the $i-$th market phase the arrival time of the $k-th$ buy (resp. sell) market order is given by $\tau^{cl}_{i-1}+\tau^{i,a}_k$ (resp. $\tau^{i,b}_k$) where $ \tau^{i,a}_k$ (resp. $\tau^{i, b}_k$) is the $k-$th event time of the counting process $(N^{i, a}_s)_{s\geq 0}$ (resp. $(N^{i,b}_s)_{s\geq 0}$). Consequently the opening time of the $i-$th auction is $\tau^{op}_i = \tau^{cl}_{i-1}+\tau^{i,a}_1 \wedge \tau^{i,b}_1$. We suppose that each market taker sends market orders with constant volume $v$. Moreover we assume that $(N^{i, mm}, N^{i,a}, N^{i,b})$ is independent of the efficient price $P$. We define $I^i$ as the cumulated imbalance of the market takers: $I^i_t = vN^{i,a}_{t} - vN^{i,b}_{t}$. The aggregated demand of the market takers at the clearing of the $i-$th auction is thus given by $I^i_{\Delta_i}$. \\

\noindent We now make the following natural assumption, which states that market is regenerative.
\begin{assumption}
\label{assumption:order_flow}
The market dynamics satisfy:
\begin{itemize}
\item[i)]After each auction clearing the market regenerates:  the processes $(N^{i, mm}, N^{i, a}, N^{i, b}, I^i)_{i\geq 0}$ are independent and identically distributed.
\item[ii)]The random variables $(\tau^{i,a}_1 \wedge \tau^{i, b}_1)_{i\geq 0}$ are i.i.d. with exponential law with parameter $\nu$.
\item[iii)]The random variables $N^{1,a}_{\tau^{cl}_1}$ and $N^{1, b}_{\tau^{cl}_1}$ are squared integrable.
\end{itemize}
\end{assumption} 

Points $i$) and $ii$) of Assumption \ref{assumption:order_flow} imply that market order flow is basically a Poisson process, which is the most standard dynamic used in the microstructure literature, see \cite{avellaneda2008high,gueant2017optimal}. This assumption is not perfectly realistic, in particular it does not enable us to reproduce the long memory property of market order flow, see for example \cite{bouchaud2009markets}. A possible way to relax this assumption would be to consider Hawkes-type intensities. However this would make the model much more intricate in terms of computation and calibration. Point $iii$) is a classical technical assumption.\\

Note that Points $i$) and $ii$) of Assumption \ref{assumption:order_flow} mean that for any $i\geq 0$, $\tau^{op}_{i+1} - \tau^{cl}_i$ follows an exponential random variable with parameter $\nu$. We also consider $(N^{mm}, N^a, N^b, I)$ a random variable with the law of $(N^{1,mm}, N^{1,a}, N^{1, b}, I^1)$. This will be useful to lighten some notations.\\

In practice it is very unlikely that a market taker sends a market order if there is no liquidity in the LOB and a situation with empty LOB is very unrealistic. A way to adapt the non empty LOB assumption setting is to consider that the first market maker always arrives before the auction clearing occurs. It means that almost surely for any $i$ we have $\tau^{i, mm}_1<(\tau^{i,a}_1\wedge\tau^{i,b}_1)+h$. Hence we consider the following assumption
\begin{assumption}
\label{assumption:market_maker}
Let $\mu>0$. The density of $\big( \tau^{1, mm}_1, (\tau^{1,a}_1\wedge\tau^{1,b}_1)\big)$ at point $(s, t)\in \mathbb{R}^2$ is given by
\begin{equation*}
\mathbf{1}_{0\leq s\leq t+h}\frac{\mu e^{-\mu s}}{1- e^{-\mu (t+h)}}\mathrm{d}s~  \nu e^{-\nu t}\mathbf{1}_{t\geq 0}\mathrm{d}t.
\end{equation*}
Finally we assume that $(N^{1, mm}_{s+\tau^{1, mm}_1} - 1)_{0 \leq s\leq h}$ is a Poisson process with intensity $\mu$ that is independent of $P$ and $(N^{1, a}_s, N^{1, b}_s)_{s\geq \tau ^{1, mm}_1}$.
\end{assumption}
Assumption \ref{assumption:market_maker} means that $(N^{1, mm}_s)_{0 \leq s\leq \tau_1^{op}+h }$ has the law of a Poisson process with intensity $\mu$ conditional on the fact that its first event occurs before time $\tau^{cl}_1$. This assumption\footnote{An alternative idea leading to a very different approach would be to endogenize the market behavior of market makers, see \cite{du2017optimal}} also allows to obtain a non-degenerate CLOB at the limit $h=0$.

\subsection{Clearing rule}
\label{subsubsec:clearing_rule}
We now explain how the clearing price is settled at the end of an auction. At time $\tau_{i}^{cl} = \tau^{op}_i+h$, a clearing price $P^{cl}_{\tau^{cl}_i}$ is set in order to maximize the exchanged volume. This clearing rule is used in most of electronic markets for the opening and clearing auctions. This is also the rule considered in the academic literature (see for instance \cite{du2017optimal}). We denote by $F^{-}(p)$ (resp. $F^{+}(p)$) the total number of shares that buyers (resp. sellers) are willing to buy (resp. sell) at price $p$. The function $F^{-}$ (resp. $F^{+}$) is decreasing (resp. increasing). Assume that a clearing price $P^{cl}$ is set. The total volume exchanged is then $F^{-}(P^{cl})\wedge F^{+}(P^{cl})$. Now suppose that $F^{-}$ and $F^{+}$ are continuous at point $P^{cl}$ and $F^{-}(P^{cl})<F^{+}(P^{cl})$. If there is still remaining liquidity on the bid side of the book (formally if $F^-$ is not bounded by $F^{-}(P^{cl})$), the exchanged volume is not optimal as it may be improved by decreasing the price. Conversely, assume that $F^{-}(P^{cl})>F^{+}(P^{cl})$ and if there is liquidity on the ask side (formally, if $F^{+}$ is not bounded by $F^{+}(P^{cl})$), then the exchanged volume is not optimal as it may be improved by increasing the price. Thus, when such equality is possible and in order to maximize volume at the clearing time, the optimal clearing price has to satisfied 
\begin{equation}\label{eq:clearingprice}
F^{-}(P^{cl})-F^{+}(P^{cl}) = 0.
\end{equation}
Note that the value $F^-(+\infty)$ (resp. $F^+(-\infty)$) is the number of shares to be bought (resp. to be sold) at any price.\\

\noindent The function $F^{-} - F^{+}$ is the algebraic supply-demand function of all market participants together (market makers and market takers). Thus we have obtained that the clearing price is a zero of the aggregated supply-demand of the agents. Consequently, in our framework, the clearing price $P^{cl}_{\tau^{cl}_i}$ of the $i-th$ auction, defined as a solution of \eqref{eq:clearingprice}, can be found solving the following equation:

$$
\sum_{k =  1 }^{ N^{i, mm}_{\Delta_i}  } S_k(p) - I^i_{\Delta_i} = 0.
$$
The $i-$th clearing price is then given by
\begin{equation}
\label{eq:clearing_price}
P^{cl}_{\tau^{cl}_i} = \frac{1}{N^{i,mm}_{\Delta_i}}\sum_{k = 1 }^{N^{i,mm}_{\Delta_{i}}} \tilde{P}_{k} + \frac{1}{K}\frac{I^i_{\Delta_{i}} }{ N_{\Delta_{i}}^{i,mm} }.
\end{equation}
Finally, we define the mid-price $P^{mid}$ of the LOB as the obtained clearing price without taking into account market orders in the auction clearing:
\begin{equation}
\label{midprice}
P^{mid}_{\tau^{cl}_i}  = \frac{1}{N^{i,mm}_{\Delta_i}}\sum_{k = 1 }^{ N^{i,mm}_{\Delta_i}  }\tilde{P}_k.
\end{equation}

\subsection{A metric for the quality of the price formation process}
\label{sec:compare_market}

One of the main roles of a financial market is to reveal with accuracy the price of the underlying asset, guaranteeing fair transaction prices to market participants. In our framework, this is equivalent to have a clearing price close to the efficient price. Therefore a natural criterion to compare different auction durations is to assess, with respect to the auctions duration, the accumulated error between the efficient price and the clearing prices over the day. To do so, we consider the following time-weighted quadratic error:
\begin{equation}\label{eq:zh}
Z^h_t = \sum_{i = 1}^{N^{cl}_t - 1} \Delta_{i+1}(P_{\tau^{cl}_i}-P^{cl}_{\tau^{cl}_i})^2 + (t - \tau^{cl}_{N^{cl}_t})(P_{\tau^{cl}_{N^{cl}_t}}-P^{cl}_{\tau^{cl}_{N^{cl}_t}})^2,
\end{equation} where $N^{cl}_t$ denotes the number of auctions cleared before time $t$. Thus, for each auction, we consider the quadratic deviation between the clearing price and the efficient price and weight this deviation by the time to wait until a new price is set. Note that \eqref{eq:zh} may be rewritten  
$$
Z^h_t = \int_0^t (\overline{P}^{cl}_s-\overline{P}_s)^2 \mathrm{d}s,
$$
where the processes $\overline{P}^{cl}_s$ and $\overline{P}_s$ are respectively the clearing and efficient price at the last auction clearing time before time $s$, that is
$$
(\overline{P}^{cl}_s, \overline{P}_s) = (P^{cl}_{ \tau^{cl}_i }, P_{\tau^{cl}_i}), \text{ where } i = {\sup}\{j\geq 1,\text{ s.t }\tau_j^{cl}\leq s \}.
$$

We define an auction duration $h^{*}$ as optimal if almost surely, $Z^{h^*}_t$ is asymptotically smaller than $Z^h_t$ for any $h\geq 0$. Using the fact that $\big( (\overline{P}^{cl}_s - \overline{P}_s)^2\big)_{s\geq 0}$ is a regenerative process we obtain, see Appendix \ref{appendix:regenerative}, the following important result for our asymptotic computations.
\begin{lemma}
\label{lemma:utility}
The following convergence holds almost surely:
$$
\underset{t \rightarrow +\infty }{\lim} \frac{Z^h_t}{t} =  \mathbb{E}[( P^{cl}_{\tau^{cl}_1} - P_{\tau^{cl}_1})^2].
$$
\end{lemma}
In light of Lemma \ref{lemma:utility}, a duration $h^*$ is optimal if it is a minimizer of the function $E$ given by
$$
E(h) = \mathbb{E}[(P^{cl}_{\tau^{cl}_1 } - P_{\tau^{cl}_1 })^2].
$$
We also consider the efficiency of the mid-price defined in \eqref{midprice}, denoted by $E^{mid}$:
$$
E^{mid}(h) = \mathbb{E}[(P^{mid}_{\tau^{op}_1 + h} - P_{\tau^{op}_1 + h})^2].
$$
We now give our first main theorem. It provides a semi-explicit expression for the function $E$. Its proof is given in Appendix \ref{appendix:proof:utility_function}.
\begin{theorem}
\label{th:utility_function}
The quality of price formation process metric satisfies:
\begin{equation*}
E(h) = E^{mid}(h) + \frac{\mathbb{E}[I^2_{\tau^{cl}_1}]}{K^2}(1 - e^{-\mu h}\frac{\nu}{\nu + \mu})^{-1} e^{\nu h } \int_h^{+\infty}\!\!\!\!\nu e^{-\nu t }e^{-\mu t}\int_0^{\mu t} \frac{1}{s}\int_0^s \frac{e^u - 1}{u} \mathrm{d}u \mathrm{d}s \mathrm{d}t,
\end{equation*}
with $E^{mid}(h)$ given by
\begin{equation}\label{Emid}
(1 - e^{-\mu h}\frac{\nu}{\nu + \mu})^{-1} e^{\nu h }\int_{h}^{+\infty}\!\!\!\!\nu e^{-\nu t } \Big( (\sigma^2_{f}\frac{t}{6}+\sigma^2)e^{-\mu t}\int_{0}^{\mu t}\frac{e^s - 1}{s}\mathrm{d}s +\sigma_f^2 \frac{t}{3}(1 - e^{-\mu t})  \Big)\mathrm{d}t.
\end{equation}
\end{theorem}
\begin{remark} Note that we can simplify the double integrals by using the so-called Exponential Integral function $E_1:\mathbb R_+^*\longrightarrow \mathbb R_+$ defined by $E_1(x)=\int_x^{+\infty} \frac{e^{-u}}udu$. We thus get
\begin{align*}
E(h) = E^{mid}(h) +& \frac{\mathbb{E}[I^2_{\tau^{cl}_1}]}{K^2}(1 - e^{-\mu h}\frac{\nu}{\nu + \mu})^{-1} e^{\nu h } \Big( \int_h^{+\infty} \frac{e^{-\mu u -1}}{u}\frac{\nu}{\nu + \mu}E_1\big( (\nu + \mu)u\big)\mathrm{d}u\\
& + \int_0^h \frac{e^{-\mu u -1}}{u}\frac{\nu}{\nu + \mu}\big( \log(\frac{h}{u})e^{-(\nu + \mu)h} + E_1((\nu+\mu)h) \big) \mathrm{d}u\Big),
\end{align*}
where $E^{mid}$ is given by \eqref{Emid}.
\end{remark}

We remark from Theorem \ref{th:utility_function} that for given $h>0$, $E(h)>E^{mid}(h)$. This is quite intuitive: the presence of market orders induces here additional deviations of clearing prices which are not directly driven by information, rather by imbalance between supply and demand. Of course when $\mu=0$, we get $E(h)=E^{mid}(h)$. We also see that the price formation process is of higher quality when $K$ is large. In that case a large amount of liquidity is already present close to the efficient price, leading to better transaction prices. Finally note that a similar expression as the one in Theorem \ref{th:utility_function} can be obtained when we allow market makers to cancel their orders, see Appendix \ref{appendix:cancellation}.\\ 

\noindent If we have access to the quantity $\mathbb{E}[I_{\tau^{op}_1 + h}^2]$, which depends on the market takers behavior, Theorem \ref{th:utility_function} enables us to compute the function $E$ and therefore to find the optimal auction duration by minimizing $E$. We can for example consider the standard assumption that $N^a$ and $N^b$ are independent Poisson processes with intensity $\nu/2$ along the auction (this is consistent with Assumption \ref{assumption:order_flow}). In this case we get
$$\mathbb{E}[I_{\tau^{op}_1 + h}^2]= v^2(\nu h+1),$$ see Appendix \ref{appendix_calculnaiveMT}. Therefore the function $h\rightarrow E(h)$ of Theorem \ref{th:utility_function} becomes fully explicit and we can obtain numerically the optimal duration. We refer to Section \ref{sec:results} for numerical details, empirical results and statistical methodology to estimate the parameters appearing in the expression of $E(h)$.\\

\noindent The Poisson assumption for the market order flow is very classical and leads to easy computations and simple results. However, in an auction setting, market orders play a crucial role and one should also investigate the possibility of strategic placements, taking into account the auction environment. We deal with this situation in the next section.

\section{Strategic market takers}
\label{sec:imbalance_process}
In practice, market orders are sent through algorithms optimizing transaction times. So, in this section, we consider that market takers aim at minimizing their trading costs by adapting their trading intensities to the market state. We formalize this into a competitive game between buying and selling market takers. We show that this game admits a Nash equilibrium. Moreover, when market takers follow the strategies corresponding to this Nash equilibrium, we can compute the key quantity $\mathbb{E}[I_{\tau^{op}_{1} + h}^2]$ appearing in the expression of $E(h)$. Note that it would of course be interesting to also consider that market makers are also strategic alter their behaviors in response to changing duration of the auction, see \cite{budish2015high,du2017optimal}. However we left this case for further research and focus here on the specific feature of auction markets from a market taker viewpoint.

\subsection{Trading costs of market takers}
\label{subsec:trading_cost_market_takers}
We model the aggregated group of buying (resp. selling) market takers as a single player called Player $a$ (resp. $b$). During the auctions, Player $a$ (resp. $b$) controls the intensity $\lambda^a$ (resp. $\lambda^b$) of the arrival process $N^a$ (resp. $N^b$), wishing to get minimal costs. In practice, market orders are often send to execute large metaorders over a specified time-interval. Consequently, whatever the market design, market takers are usually required to buy or sell a certain volume on a given period. To reproduce the fact that market takers intensities can neither be too high nor too low, we assume that $\lambda^a$ and $\lambda^b$ are bounded from above and below by two positive constants $\lambda_{+}$ and $\lambda_{-}$.\\

The aggregated total trading cost at time $t$ of buying market takers, denoted by $C^a_t$, satisfies
$$
C^a_t = \sum_{i = 1}^{N^{cl}_t} N^{i,a}_{\Delta_{i}}(P^{cl}_{\tau^{cl}_{i}} - P_{\tau^{cl}_{i}}).
$$
From Theorem 3.1 in Chap VI in \cite{asmussen2008applied} together with the fact that the market is regenerative we obtain the following lemma on the asymptotic behavior of the trading costs.
\begin{lemma}
\label{lemma:trading_cost}
We have the following almost sure convergence:
$$
\underset{t \rightarrow + \infty}{\lim}\frac{C^a_t}{t} = \mathbb{E}[N^a_{\tau^{cl}_1}(P^{cl}_{\tau^{cl}_{1}} - P_{\tau^{cl}_{1}})]\frac{\nu}{1+\nu h}.
$$
\end{lemma}
Therefore, in the long run, the average trading cost of buying market takers is a multiple of
$$
\mathbb{E}[N^a_{\tau^{cl}_1}(P^{cl}_{\tau^{cl}_{1}} - P_{\tau^{cl}_{1}})]  = \frac{v}{K}\mathbb E[\frac{1}{N^{mm}_{\Delta_i}}] \mathbb{E}[N^a_{\tau^{cl}_1}(N^a_{\tau^{cl}_1} - N^b_{\tau^{cl}_1})].
$$
Now writing $N^a_{\tau^{cl}_1} = N^a_{\tau^{op}_1+h} - N^a_{\tau^{op}_1} + N^a_{\tau^{op}_1}$ and using the fact that $N^a_{\tau^{op}_1}$ is either equal to one or zero, solving the problem of Player $a$ is equivalent to be able to minimize
$$
\mathbb{E}[N^a_h(N^a_h-N^b_h)]
$$
when $(N^a_0, N^b_0) = (1, 0)$ and when $(N^a_0, N^b_0) = (0, 1)$. Consequently, for any $(\alpha,\beta)\in \mathbb N^2$, we consider the more general problem for Player $a$ minimizing $$\mathbb{E}[N^a_h(N^a_h-N^b_h)|(N^a_0, N^b_0) = (\alpha,\beta)].$$ In the same way, Player $b$ minimizes $\mathbb{E}[N^b_h(N^b_h-N^a_h)|(N^a_0, N^b_0) = (\alpha,\beta)]$ . Each player aims at deriving its own trading intensity which will lead to the smallest possible trading costs for him.\\

\noindent Note that in our setting, Assumption \ref{assumption:order_flow} implies that market takers reset their strategies at the beginning of each auction. We could have considered the case where market takers optimize their behavior all along the day. However, since we are interested in the effects of auction durations in a stationary context, our framework remains reasonable.

\subsection{Nash equilibrium}
We now give our result on the existence of a Nash equilibrium in this game of competing market takers. We consider that market takers control their trading intensities. The set of admissible controls is denoted by $\mathcal{U}$ and defined as the set of $\mathbb F-\text{predictable processes}$ with values in $[\lambda_-, \lambda_{+}]$ for fixed $0<\lambda_-\leq \lambda_+$. Any couple of strategies $(\lambda_a,\lambda_b) \in \mathcal U^2$ of Player $a$ and $b$ induces a probability measure $\mathbb{P}^{\lambda_a, \lambda_b}$ such that
$$
N^a_\cdot - \int_0^\cdot \lambda^a_s\mathrm{d}s~\text{and}~N^b_\cdot - \int_0^\cdot \lambda^b_s\mathrm{d}s$$ are martingales under $\mathbb{P}^{\lambda_a, \lambda_b}$. In order to minimize its costs, Player $a$ solves
	\begin{equation}
	\label{eq:buyer_pb}
		 \inf_{\lambda^a \in \mathcal U} V^{a, \alpha,\beta}_h(\lambda_a,\lambda_b),
	\end{equation}
	with $V^{a, \alpha,\beta}_h(\lambda_a,\lambda_b)= \mathbb{E}^{\mathbb{P}^{\lambda_a, \lambda_b}}[N^a_h(N^a_h - N^b_h)|(N^a_0, N^b_0) = (\alpha,\beta)],$
	for fixed $\lambda_b$ chosen by the selling market taker, Player $b$. Symmetrically, Player $b$ solves
	\begin{equation}
	\label{eq:seller_pb}
	\inf_{\lambda^b\in \mathcal U}V^{b,\alpha,\beta}_h(\lambda_a,\lambda_b),
	\end{equation}
	with $V^{b,\alpha,\beta}_h(\lambda_a,\lambda_b)= \mathbb{E}^{\mathbb{P}^{\lambda_a, \lambda_b}}[N^b_h(N^b_h - N^a_h)|(N^a_0, N^b_0) = (\alpha,\beta)]$
	for fixed $\lambda_a$ chosen by the buying market taker, Player $a$. A Nash equilibrium is obtained if the two optimization problems \eqref{eq:buyer_pb} and \eqref{eq:seller_pb} can be addressed simultaneously.\\

Note that this framework is realistic regarding the information observable by market takers. Indeed we only assume that market takers observe market orders imbalance. This information is for example available on the Euronext platform for the opening and closing auctions and on the auctions market of BATS-Cboe.\\

\noindent	We prove that this game indeed admits a (non-necessarily unique) Nash equilibrium with corresponding optimal controls $(\lambda_a^\star, \lambda_b^\star)$. More precisely using these notations we have the following result.
\begin{theorem}
\label{theo:imbalance_nash_equilibrium}
There exists a Nash equilibrium to the simultaneous optimization problem \eqref{eq:buyer_pb}-\eqref{eq:seller_pb} given by some Markovian controls\footnote{The notion of Markovian control has to be understood in the sense of \cite[Definition 2.10]{carmona2018probabilistic}} $(\lambda_a^\star, \lambda_b^\star)$ satisfying
	$$
  \inf_{\lambda^a \in \mathcal U} V^{a, \alpha,\beta}_h(\lambda_a,\lambda_b^\star)=  \mathbb{E}^{\mathbb{P}^{\lambda_a^\star, \lambda_b^\star}}[N^a_h(N^a_h - N^b_h)|(N^a_0, N^b_0) = (\alpha,\beta)]$$
	\text{ and }	$$  \inf_{\lambda^b \in \mathcal U} V^{b, \alpha,\beta}_h(\lambda_a^\star,\lambda_b)=  \mathbb{E}^{\mathbb{P}^{\lambda_a^\star, \lambda_b^\star}}[N^b_h(N^b_h - N^a_h)|(N^a_0, N^b_0) = (\alpha,\beta)].
	$$
\end{theorem}

\noindent The proof of Theorem \ref{theo:imbalance_nash_equilibrium} is provided in Appendix \ref{appendix:existence_nash_equilibrium}. The HJB equation related to the optimization problem is somehow degenerated. This prevents us from using classical arguments to obtain a solution. In order to give intuition about it we give here a short sketch of the proof.
\begin{itemize}
\item[Step 0.] We first consider a smoothed version of the HJB equation associated with our problem. Hence, the proof of Theorem \ref{theo:imbalance_nash_equilibrium} is reduced to the existence and then convergence of the solutions of a (smooth) system of HJB equations (see Theorem \ref{thm:edpapproach}).
\item[Step 1.]We consider the backward stochastic differential equation (BSDE for short) associated to the smoothed HJB equation. The existence of a Nash equilibrium is then related to the existence of a solution to this (Lipschitz) BSDE.
\item[Step 2.] We prove that the sequence of BSDEs converges in suitable spaces towards a solution of a degenerate BSDE. 
\item[Step 3.] We conclude by showing that the solution we obtain at the limit corresponds to a Nash equilibrium of the competition between market takers.
\end{itemize}

Note that we do not get uniqueness of the Nash equilibrium, only the existence. Since the generator of the BSDE associated to this problem has discontinuities there is almost no chance that a uniqueness result can be found by classical methods. Moreover, even if the method used give uniqueness of the limit Nash equilibrium, this limit will be strongly dependent of the smoothing procedure.\\

The proof of Theorem \ref{theo:imbalance_nash_equilibrium} also provides a numerical method to approximate $V^{a,1,0}_h (\lambda_a^\star,\lambda_b^\star)$ and $V^{b,1,0}_h (\lambda_a^\star,\lambda_b^\star)$ using solutions of some integro-differential equations, see Appendix \ref{appendix:numerical_nash}. It is particularly important since it enables us to compute optimal auction durations when market takers are playing the Nash equilibrium. This is because the function $E$ of Theorem \ref{th:utility_function} explicitly depends on $V^{a,1,0}_h (\lambda_a^\star,\lambda_b^\star)$ and $V^{b,1,0}_h (\lambda_a^\star,\lambda_b^\star)$, as stated in the following corollary.
\begin{corollary}\label{cor:Eh}
Under the Nash equilibrium $(\lambda_a^{\star},\lambda_b^\star)$, we have
	$$
	\mathbb{E}[I_{\tau^{op}_1 + h}^2] = V^{a,1,0}_h (\lambda_a^\star,\lambda_b^\star)+ V^{b,1,0}_h (\lambda_a^\star,\lambda_b^\star).
	$$
	\end{corollary}

\section{Optimal auction durations for some European stocks}
\label{sec:results}

We give here the results obtained on real data when applying our methodology to derive optimal auction durations. We consider both situations of non-strategic and strategic market takers and compare with the CLOB case.

\subsection{Description of the data}
\label{subsec:description_data}

We have access to intra-day market data for $77$ of the most liquid stocks traded on Euronext exchange, for all trading days of September 2018. For each stock, every trade is reported with the following information:
\begin{itemize}
\item Timestamp of the trade.
\item Traded volume.
\item Execution price.
\item Best bid and ask prices just before the transaction.
\item Volumes at best bid and best ask just before the transaction.
\end{itemize}
We discard from our study trades related to $1\%$ upper and lower quantiles in term of volume in order to remove some outliers.

\subsection{Calibration of model parameters}
\label{subsec:calibrate_model}
Our market data are CLOB data and not auction data. Still, we are able to calibrate the parameters of our model as explained below.

\subsubsection{Market takers parameters}
\label{subsubsec:market_takers_parameters}
The behavior of market takers is characterized by three parameters:
\begin{itemize}
\item Their intensity of arrival between two auctions $\nu$.
\item The volume of market orders $v$.
\item The upper and lower bounds for their trading intensity $\lambda_-$ and $\lambda_+$.
\end{itemize}
CLOB corresponds to the case where auctions last zero second. Consequently, in our framework, the market order flow in a CLOB market is given by two Poisson processes $N^a$ and $N^b$ with intensity $\nu$. Thus we estimate $\nu$ by the average number of market orders per day divided by the duration of a trading day and $v$ by the average volume of a market order. Finally we set $\lambda_{+}=2\nu$ and $\lambda_{-} = \nu / 4$. This choice seems reasonable since the market order flow should have similar order of magnitude irrespectively of the market design so that agents can complete execution of their metaorders.

\subsubsection{Market makers parameters and calibration of price volatility }
\label{subsubsec:market_makers_parameters}
The behavior of market makers is characterized by three parameters:
\begin{itemize}
\item The variance $\sigma$ of the $(g_i)_{i\geq 0}$. We assume that $\sigma$ is equal to the implicit spread of the asset that we estimate from the uncertainty zones model of \cite{dayri2015large}.
\item The intensity of market makers arrivals $\mu$.
\item The slope of their supply function $K$.
\end{itemize}
Let $\alpha$ be the tick value of the asset. According to our model, in the CLOB case, the average volume available in the first limit of the LOB when a market taker arrives, denoted by $e$, satisfies 
$$
e = K \alpha \mathbb{E}[ N^{mm}_{\tau^{op}_1}] = K \alpha \frac{\nu + \mu}{\nu}
$$
and the average squared volume of the first limit, denoted by $\varsigma$, satisfies
$$
\varsigma = K^2 \alpha^2 \mathbb{E}[ (N^{mm}_{ \tau^{op}_1 })^2] = K^2 \alpha^2 \frac{\nu + \mu}{\nu} (1 + 2\frac{\mu}{\nu}).
$$
Those results are a direct consequence of Assumption \ref{assumption:market_maker} and of some computations. Consequently we have
$$
K =  (2e- \frac{\varsigma }{e})\alpha^{-1} \text{ and } \mu = \nu ( \frac{e}{\alpha K}- 1).
$$
So we can estimate $\mu$ and $K$ from empirical measurements of $e$ and $\varsigma$. Finally, we estimate the volatility $\sigma_f$ of the efficient price from the five minutes sampling based realized volatility of the traded price.

\subsection{Numerical results}
\label{subsec:numerical_results}

Using our approach, we provide in Table \ref{tab:results_a} and \ref{tab:results_b} the optimal auction durations for 77 stocks traded on Euronext. We give the results when assuming Poisson arrivals for the market takers and when considering they optimize their trading costs, leading to a Nash equilibrium (see Appendix \ref{appendix:numerical_nash} for numerical aspects in this case).\\

The first column is the stock name. In the Poisson (resp. Nash) case, the second (resp. fourth) column is the optimal duration in seconds. The third (resp. fourth) one is the relative difference of quality of the price formation process between the optimal duration case and the CLOB situation: $(E(0)-E(h^*))/E(h^*)$. In the optimal durations columns we provide estimated optimal durations together with $90\%$ confidence interval (with respect to the estimated value for the parameter $\nu$).\\

\begin{table}[]
\centering
\begin{tabular}{l|c|c|c|c}
                    &                    DurationPoisson&                DiffrelPoisson       &                    DurationNash&           DiffrelNash\\
\hline
Bouygues            &                    228 [226;230]   &                1\%           &              152 [150;153] &              20\%\\
Arkema              &                    397 [392;400]   &                23 \%         &              268 [265; 272]&              19 \%\\
Michelin            &                    1053 [1046;1060]&                60\%          &              763 [757;768] &              89\%\\
Eurofins Scient.    &                    761 [749;773]   &                18\%          &              554 [546;563] &              37\%\\
Engie               &                    866 [857;875]   &                104\%         &              866 [857;875] &              158\%\\
\hline
Stmicroelectronics  &                    177 [176;179]   &                2\%           &              123 [122;124] &              21\%\\
Alstom              &                    0 [0;0]         &                0\%           &              180 [178;181] &              14\%\\
Legrand SA          &                    325 [322;329]   &                0\%           &              216 [214;221] &              19\%\\
Eiffage             &                    0 [0;0]         &                0\%           &              149 [147;150] &              12\%\\
Eramet              &                    1086 [1074;1098]&                30\%          &              812 [803;822] &              50\%\\
\hline
SES Sa              &                    0 [0;0]         &                0\%           &              81 [80;83]    &              6\%\\
Pernod Ricard       &                    427 [423;430]   &                22\%          &              301 [298;304] &              45\%\\
Iliad               &                    163 [162;164]   &                0\%           &              109 [108;110] &              18\%\\
Faurecia            &                    0 [0;0]         &                0\%           &              36 [35;37]    &              4\%\\
Orange              &                    382 [379;385]   &                21\%          &              274 [273.6;278] &              42\%\\
\hline
Sodexo              &                    0 [0;0]         &                0\%           &              49 [51;47]    &              1\%\\
Air France - KLM    &                    295 [292;297]   &                17\%          &              218 [216;220] &              35\%\\
Teleperformance     &                    1241 [1224;1259]&                27 \%         &              881 [868;894] &              50 \%\\
Hermes              &                    295 [292;298]   &                1\%           &              205 [203;207] &              19\%\\
Eutelsat Com.       &                    0 [0; 0]        &                0 \%          &              40 [39; 42]   &              2 \%\\
\hline
Nexans              &                    487 [480;494]   &                8\%           &              360 [356;365] &              23\%\\
Ingenico Group      &                    0 [0;0]         &                0\%           &              143 [142;144] &              15\%\\
Unibail - Wfd Unibai&                    187 [186;188]   &                19\%          &              142 [141;143] &              36\%\\
Plastic Omnium      &                    0 [0;0]         &                0\%           &              176.5 [176.3;176.8] &              9\%\\
Veolia Environ.     &                    350 [346;353]   &                3\%           &              253 [251;256] &              21\%\\
\hline
Schneider Electric  &                    246 [245;248]   &                39\%          &              171 [170;172] &              65\%\\
Peugeot             &                    386 [383;389]   &                10\%          &              282 [280;285] &              29\%\\
Vinci               &                    350 [348;353]   &                39\%          &              252 [250;253] &              64\%\\
CGG                 &                    837 [827;847]   &                15\%          &              605 [597;613] &              36\%\\
Atos                &                    962 [954;969]   &                66\%          &              700 [694;706] &              95\%\\
\hline
Suez Environnement  &                    0 [0;0]         &                0\%           &              311 [308;315] &              14\%\\
Danone              &                    204 [203;206]   &                15\%          &              146 [145;147] &              35\%\\
Kering              &                    133 [132;134]   &                19\%          &              93.4 [93.1;94]    &              42\%
\end{tabular}
\caption{Optimal auction durations (in seconds) Part 1 with a $90\%$ confidence interval.}
\label{tab:results_a}
\end{table}

\begin{table}[]
\centering
\begin{tabular}{l|c|c|c|c}
                    &                    DurationPoisson&                DiffrelPoisson       &                    DurationNash&           DiffrelNash\\
\hline
EssilorLuxottica    &                    342 [339;345]   &                30\%          &              238 [236;240] &              55\%\\
Lagardere           &                    0 [0;0]         &                0 \%          &              42 [39; 44]   &              3 \%\\
Credit Agricole     &                    87.7 [87.2;88.5]      &                2\%           &              58.6 [58;59.4]    &              22\%\\
CapGemini           &                    502 [497;508]   &                20\%          &              354 [350;358] &              43\%\\
Lvmh                &                    121 [120;122]   &                6\%           &              87.3 [87;88]    &              25\%\\
\hline
Valeo               &                    0 [0;0]         &                0\%           &              98 [97;98.2]    &              16\%\\
Air Liquide         &                    627 [622;632]   &                35\%          &              459 [456;463] &              58\%\\
Total               &                    359 [357;360]   &                60\%          &              261 [260;263] &              89\%\\
Vivendi             &                    1023 [1014;1031]&                42\%          &              750 [743;756] &              67\%\\
Casino Guichard     &                    158 [157;159]   &                15\%          &              119 [118;120] &              28\%\\
\hline
Societe Generale    &                    104 [104;105]   &                18\%          &              74.1 [74;74.3]    &              40\%\\
Klepierre           &                    0 [0;0]         &                0\%           &              219 [217;221] &              14\%\\
Publicis Groupe     &                    601 [595;606]   &                32\%          &              428 [424;432] &              56\%\\
Sanofi              &                    124 [123;124]   &                12\%          &              88.2 [88;89]    &              32\%\\
Thales              &                    644 [637;652]   &                23\%          &              454 [449;460] &              46\%\\
\hline
TechnipFMC          &                    331 [327;334]   &                7\%           &              234 [232;236] &              27\%\\
Bnp Paribas         &                    104.3 [104.2;104.8]   &                18\%          &              73.4 [73.2;74]    &              41\%\\
Safran              &                    0 [0;0]         &                0\%           &              107 [106;108] &              16\%\\
Saint Gobain        &                    0 [0;0]         &                0\%           &              58.2 [58;59]    &              11\%\\
Orpea               &                    834 [822;846]   &                29\%          &              578 [569;587] &              55\%\\
\hline
Carrefour           &                    410 [407;413]   &                34\%          &              293 [291;295] &              58\%\\
Ipsen               &                    827 [817;838]   &                65\%          &              551 [544;559] &              101\%\\
Natixis             &                    351 [348;354]   &                9\%           &              253 [251;255] &              28\%\\
EDF                 &                    341 [338;344]   &                15\%          &              246 [244;248] &              35\%\\
Axa                 &                    252 [251;254]   &                36\%          &              182 [181;183] &              60\%\\
\hline
Dassault Systemes   &                    316 [312;319]   &                7\%           &              222 [220;224] &              27\%\\
Accor Hotels        &                    0 [0;0]         &                0\%           &              105.3 [105.8;104.7] &              6\%\\
Airbus              &                    210 [209;211]   &                34\%          &              146 [145;147] &              60\%\\
Ubi Soft Entertain  &                    0 [0;0]         &                0\%           &              43.4 [43;44]    &              1\%\\
Renault             &                    0 [0;0]         &                0\%           &              41.7 [41;42.2]    &              3\%\\
\hline
Solvay              &                    528 [522;534]   &                11\%          &              375 [371;380] &              32\%\\
Edenred             &                    313 [309;316]   &                8\%           &              210 [208;212] &              29\%
\end{tabular}
\caption{Optimal auction durations (in seconds) Part 2 with a $90\%$ confidence interval.}
\label{tab:results_b}
\end{table}

The optimal duration range is essentially between $0$ and $10$ minutes and our results are very robust to the parameter $\nu$. For all the assets such that the optimal auction duration for Poisson market takers is positive, the optimal duration in the Nash case is smaller. Some assets have the CLOB structure as optimal in the Poisson case. However, when considering the Nash case, CLOB become always suboptimal. We also remark that no straightforward structural explanation (sector, capitalization, ...) seems to explain the difference in optimal duration between assets. Finding microstructural foundations for these results is left for further work.\\

As explained in Section \ref{sec:imbalance_process}, we constrain market takers trading intensities to the range $[\lambda_{-}, \lambda_{+}]$. From numerical experiments, by testing several ranges of controls $[\lambda_{-}, \lambda_{+}]$, we have observed that the optimal duration is quite robust to those parameters. Still we remark the following sensitivities: if we allow for a smaller  $\lambda_{-}$ without modifying $\lambda_{+}$, the optimal auction duration becomes larger. This is because having a small $\lambda_{-}$ means that market takers can send less market orders when the situation is not in their favor. This implies that $\mathbb{E}[I_{\tau^{op} + h}^2]$ increases slowlier with $h$. This moves the minimum of $E$ to a higher level. For symmetric reasons, if we raise $\lambda_{+}$ without changing $\lambda_{-}$, the optimal auction duration becomes smaller.\\

We notice that CLOBs are sometimes optimal in the Poisson case. When they are not, the difference in the values of the metric for $h=0$ and $h=h^*$ is typically not very large. Therefore even though CLOB markets are usually sub-optimal, they are in general leading to a fairly satisfactory market microstructure. On BATS-Cboe the auction duration is approximately $100 \text{ms}$ which is very small compared to the typical optimal auction durations we find. Moreover according to the empirical study \cite{besson2019benefits}, there is essentially only one market order involved in each auction. This means that the duration of auctions chosen by BATS-Cboe does not allow buyers and sellers market takers to match their orders, to the profit of market makers. Indeed a larger auction duration may lead to smaller gains for market makers. For example if market takers always match their orders with other market takers, market makers never collect the spread. Hence it is possible that BATS-Cboe chose this short duration in order to keep its platform attractive for market makers, which guarantee its liquidity. This is actually another possible point of view on this problematic that we have not considered in this paper. It is also likely that exchanges may be reluctant to change drastically their market design so that clients are not too surprised. This could also explain why they decided to move only slightly from the CLOB system.

\section{Policy implications and financial insights}

The main take away of our analysis is that \textit{one size does not fit all}: first in the spirit of \cite{budish2015high} we confirm that the nowadays almost universal CLOB mechanism (for liquid assets) may be suboptimal. We indeed show that auctions are quite often preferable in term of market quality. Second, the auction duration has to depend on some fundamental parameters of the considered asset (such as liquidity and volatility). Our work also underlines the crucial need of thorough quantitative analysis as a preliminary task to market structure modification. Here such an analysis enables us to compute optimal frequencies for auctions. We find that they are diverse but that reasonable order of magnitude is of a few minutes, which is probably fast enough for large investors. This is in contrast with the intuitive idea that in fast electronic markets auction frequency should be necessarily very high. Such a result is due to the fact that our criterion aims at finding the best price discovery mechanism and therefore we somehow take the investor point of view. Doing so, our philosophy is to build volume rather than speed driven market. This pushes forward the idea of a debate between exchanges and the various market participants in order to revisit market microstructure.\\

Note that in the CLOB case, usual parameters that the exchange can adjust to improve market microstructure are tick sizes and make-take fees. Given the order of magnitude of the auctions durations we find, the tick size will probably play a minor role in the market dynamics. However, the effects of a make-take fees schedule are yet to be investigated, as well as the impact of auctions in a situation with multiple competing exchanges. \\

We insist on the fact that our results do not advocate for the disappearance of CLOB markets. We indeed show that they are optimal for some assets and often not so far from optimality in terms of market quality. Going further in the idea that \textit{one size does not fit all}, we could think of alternating periods of CLOB and auction market within the same trading day. We currently investigate the relevance of such mechanism as a natural next step to the present paper.
	\clearpage
	
	\appendix

	\section{Proof of Theorem \ref{th:utility_function}}
	\label{appendix:proof:utility_function}

	We are reduced to compute :
	$$
	E(h) = \mathbb{E}[(P_{\tau^{cl}_1} - P^{cl}_{\tau^{cl}_1})^2|\widetilde{\Omega}]
	$$
When $N^{mm}$ be a Poisson process with intensity $\mu$ and $\widetilde\Omega = \{ N^{mm}_{\tau_1^{cl}} >0 \}.$ We are reduced to compute 
	$$
	E(h) = \mathbb{E}[(P_{\tau^{cl}_1} - P^{cl}_{\tau^{cl}_1})^2|\Omega].
	$$
	Thus, recalling that $\tau_1^{op}+h=\tau_1^{cl}$, we get
	\begin{align}
\nonumber	E(h) &= \mathbb{E}\Big[\big( \sum_{k = 1}^{N^{mm}_{\tau^{cl}_{1}}}  \frac{P_{\tau^{mm}_k} - P_{\tau^{cl}_1}}{N^{mm}_{\tau^{cl}_1}} \big)^2\big|\widetilde\Omega\Big] +  \mathbb{E}\Big[\big( \sum_{k = 1}^{N^{mm}_{\tau^{cl}_{1}}}  \frac{g_k}{N^{mm}_{\tau^{cl}_1}} \big)^2 \big|\widetilde\Omega\Big]+\frac{1}{K^2} \mathbb{E}\Big[\frac{I_{\tau^{cl}_1}^2 }{ {N^{mm}_{\tau^{cl}_1}}^2  }\big|\widetilde\Omega\Big]\\ 
\nonumber	&= \mathbb{P}(N^{mm}_{\tau^{cl}_1}>0)^{-1} \Big( \mathbb{E}\Big[\mathbf{1}_{N^{mm}_{\tau^{cl}_1}>0}\Big\{\big( \sum_{k = 1}^{N^{mm}_{\tau^{cl}_{1}}}  \frac{P_{\tau^{mm}_k}- P_{\tau^{cl}_1}}{N^{mm}_{\tau^{cl}_1}} \big)^2+\big( \sum_{k = 1}^{N^{mm}_{\tau^{cl}_{1}}}  \frac{g_k}{N^{mm}_{\tau^{cl}_1}} \big)^2+\frac{1}{K^2} \frac{I_{\tau^{cl}_1}^2 }{ {N^{mm}_{\tau^{cl}_1}}^2  }\Big\}\Big] \Big)\\
	\label{eq:utility_a}	&=\mathbb{P}(N^{mm}_{\tau^{cl}_1}>0)^{-1}e^{\nu h } \int_h^{+\infty } \nu e^{-\nu t}\big( g(t)+\sigma^2 f(t)+\frac{1}{K^2}  \ell(t)\big) \mathrm{d}t	\end{align}
	with
	\begin{equation*}
	\label{eq:utility_b}
	g(t) = \mathbb{E}[ \mathbf{1}_{N^{mm}_{t}>0}( \sum_{k = 1}^{N^{mm}_{t}}  \frac{ P_{\tau^{mm}_k} - P_{t }}{N^{mm}_{t}} )^2 ], \; f(t) = \mathbb{E}[\frac{\mathbf{1}_{N^{mm}_{t}>0}}{N^{mm}_{t}}],\; \text{ and } \ell(t) =\mathbb{E}[I_{\tau_1^{cl}}^2]  \mathbb{E}[\frac{\mathbf{1}_{N^{mm}_{t}>0}}{ {N^{mm}_{t}}^2  }].
\end{equation*}
A direct computation gives
\begin{equation}
\label{eq:proof_proba}
\mathbb{P}(N^{mm}_{\tau_1^{cl}}>0) = 1 - e^{-\mu h }\frac{\nu}{\nu + \mu}.
\end{equation}
We now turn to the computation of the function $g$. We have the following lemma.
\begin{lemma}\label{lemma_g} We have for any $t>0$	\begin{equation}
	\label{eq:utility_c}
	g(t) = \sigma_f^2 \frac{t^2}{2} \mu \mathbb{E}[  \frac{ 1}{ (N^{mm}_t + 1)^2} ] + \sigma_f^2\frac{t^3}{3}\mu^2  \mathbb{E}[\frac{ 1 }{(N^{mm}_t+ 2)^2}  ].
	\end{equation}
	
	\end{lemma}
	
	\begin{proof} Note that
	\begin{align*}
	g(t) = \sigma_f^2 \mathbb{E}\Big[ \mathbf{1}_{N^{mm}_{t}>0} \sum_{k = 1}^{N^{mm}_{t}}  \Big( \frac{ W_{\tau^{mm}_k}- W_{t}}{N^{mm}_{t}-1 + 1} \Big)^2   \Big] +\sigma_f^2 \mathbb{E}\Big[ \mathbf{1}_{N^{mm}_{t}>0} \sum_{k,l = 1\text{ s.t. }k\neq l}^{N^{mm}_{t}}   \frac{ (W_{\tau^{mm}_k} - W_{t })(W_{\tau^{mm}_l} - W_{t }) }{(N^{mm}_{t}-2 + 2)^2}\Big].
	\end{align*}
	Consider $X_t$ the Poisson scatter made of the event times of $N^{mm}$ between time $0$ and $t$. Then we have
	$$
	g(t) = \sigma_f^2 \mathbb{E}\Big[  \sum_{x \in X_t }  \frac{ (W_{x} - W_{t})^2}{ (\#\{X_t\backslash \{x\} \} + 1)^2}\Big] +\sigma_f^2 \mathbb{E}\Big[  \sum_{x,y \in X_t \text{ s.t. }x\neq y}  \frac{ (W_{x} - W_{t})(W_{y} - W_{t}) }{(\#\{X_t\backslash \{x, ~y\} \} + 2)^2} \Big].
	$$
	Since $P_t = \sigma_f W_t$ is independent of $N^{mm}$, we get
	$$
	g(t) = \sigma_f^2 \mathbb{E}\Big[  \sum_{x \in X_t }  \frac{ (t -x)^2}{ (\#\{X_t\backslash\{x\} \} + 1)^2}\Big] + \sigma_f^2 \mathbb{E}\Big[  \sum_{x,y \in X_t \text{ s.t. }x\neq y}  \frac{ (t - x) \wedge (t-y) }{(\#\{X_t\backslash \{x, ~y\} \} + 2)^2}\Big].	
	$$
	Finally using Palm's Formula, see for example \cite{coeurjolly2017tutorial}, we get
	\begin{equation*}
	g(t) = \sigma_f^2 \mathbb{E}\big[  \frac{ 1}{ (N^{mm}_t + 1)^2} \big]  \int_0^t (t-u)\mu \mathrm{d}u+ \sigma_f^2\mathbb{E}\big[\frac{ 1 }{(N^{mm}_t+ 2)^2}  \big]\int_0^t \int_0^t (t-u)\wedge (t-v)\mu^2 \mathrm{d}u \mathrm{d}v,
	\end{equation*}
and \eqref{eq:utility_c} follows.

\end{proof}

%
%
To compute explicitly $f$, $\ell$ and $g$ from Lemma \ref{lemma_g}, we need the following additional results.
\begin{lemma}
	\label{lemma:poisson_related_formulas_b}
	Let $N$ be a general inhomogeneous Poisson process with intensity measure $\lambda$. The following equalities hold:
		\begin{equation}\label{ineg_a_lem}
			\mathbb{E}[\frac{\mathbf{1}_{N_{t}>0}}{ N_t}]= e^{-m_t} \int_0^{m_t} \frac{e^{s} - 1}{s} \mathrm{d}s,\; \text{ and }\; 
			\mathbb{E}[\frac{\mathbf{1}_{N_{t}>0}}{ N^2_t}]= e^{-m_t} \int_0^{m_t} \frac{1}{s} \int_0^{s} \frac{e^{u} - 1}{u} \mathrm{d}u \mathrm{d}s,
		\end{equation}
		
	\begin{equation}\label{ineg_b_lem}
			\mathbb{E}[\frac{1}{ (1 + N_t)^2}]= \frac{e^{-m_t}}{m_t} \int_0^{m_t} \frac{e^s - 1}{s}\mathrm{d}s,\; \text{ and }\;
			\mathbb{E}[\frac{1}{ (2 + N_t)^2}] = \frac{1}{m_t^2
} \big( 1 - e^{-m_t} -e^{-m_t}\int_0^{m_t}\frac{e^s - 1}{s} \mathrm{d}s \big),
		\end{equation}
	with $m_t = \int_0^t \lambda(\mathrm{d}s).$

	\end{lemma}

	\begin{proof}[Proof of \eqref{ineg_a_lem}.]
	Note that
	$$
	\mathbb{E}[\frac{\mathbf{1}_{N_{t}>0}}{ N_t}] = \sum_{n = 1}^{+\infty} \frac{1}{n}\frac{m_t^n}{n!}e^{-m_t}\text{ and }
	\mathbb{E}[\frac{\mathbf{1}_{N_{t}>0}}{ N_t^2}] = \sum_{n = 1}^{+\infty} \frac{1}{n^2}\frac{m_t^n}{n!}e^{-m_t}.
	$$
	The functions $e_1$ and $e_2$ defined by
	$$
	e_1(x) = \sum_{n = 1}^{+\infty} \frac{1}{n}\frac{x^n}{n!} \text{ and } 
	e_2(x) = \sum_{n = 1}^{+\infty} \frac{1}{n^2}\frac{x^n}{n!}
	$$
	are continuously differentiable function, so that
	$$
	e_1'(x) = \sum_{n = 1}^{+\infty} \frac{x^{n-1}}{n!} = \frac{e^{x}-1}{x} \text{ and } xe_2'(x) = \sum_{n = 1}^{+\infty} \frac{1}{n}\frac{x^{n}}{n!} = e_1(x).
	$$
	By integrating these functions, we get \eqref{ineg_a_lem}.\\
	
	\textit{Proof of \eqref{ineg_b_lem}.}
Note that
	$$
	\mathbb{E}[\frac{1}{ (1 + N_t)^2}] = \sum_{n = 0}^{+\infty} \frac{1}{(1 + n)^2} \frac{m_t^n}{n!}e^{-m_t}\text{ and }
	\mathbb{E}[\frac{1}{ (2 + N_t)^2}] = \sum_{n = 0}^{+\infty} \frac{1}{(2 + n)^2} \frac{m_t^n}{n!}e^{-m_t}.
	$$
	Consider, for $i>0$, the functions
	$$
	r_i(x)=\sum_{n=0}^{+\infty} \frac{1}{(i + n)^2}\frac{x^{n+i}}{n!} \text{ and } s_i(x)=\sum_{n=0}^{+\infty} \frac{1}{i + n}\frac{x^{n+i}}{n!}.
	$$
	We have
	$$
	r_i'(x)=\sum_{n=0}^{+\infty} \frac{1}{i + n}\frac{x^{n+i-1}}{n!}\text{ hence }r_i(x) = \int_0^{x} \frac{s_i(s)}{s} \mathrm{d}s.
	$$
	Since
	$$
	s_i'(x)=\sum_{n=0}^{+\infty} \frac{x^{n+i-1}}{n!} = x^{i-1} e^{x}\text{ we get } r_i(x) = \int_0^{x} \frac{1}{s}\int_0^s u^{i-1}e^u \mathrm{d}u \mathrm{d}s.
	$$
	Taking $i=1$ and $i=2$ we get \eqref{ineg_b_lem}.
	\end{proof}

Injecting Equations \eqref{ineg_a_lem} and \eqref{eq:proof_proba} into $f$ and $\ell$ and Equation \eqref{ineg_b_lem} into $g$ in view of \eqref{eq:utility_c}, using \eqref{eq:utility_a} we obtain the formulas stated in Theorem \ref{th:utility_function}.

\section{Computation of the expected square imbalance in the Poisson case}
\label{appendix_calculnaiveMT}

We want to compute $\mathbb E[I^2_{\tau_1^{op}+h}]$ when $N^a$ and $N^b$ are independent Poisson processes with intensity $\nu/2$. We have
\begin{align*}
\mathbb E[I^2_{\tau_1^{op}+h}]&=v^2\mathbb E[\big((N^a_{\tau_1^{op}+h}-N^a_{\tau_1^{op}}+N^a_{\tau_1^{op}} ) - (N^b_{\tau_1^{op}+h}-N^b_{\tau_1^{op}}+N^b_{\tau_1^{op}} )\big)^2 ].
\end{align*}
Using the strong Markov property of Poisson process and taking conditional expectation with respect to $\tau_1^{op}$ we get
\begin{align*}
\mathbb E[I^2_{\tau_1^{op}+h}]&= v^2(\nu h + 1),
\end{align*}
where we use $\mathbb E[N^a_{\tau_1^{op}}]= \mathbb E[(N^a_{\tau_1^{op}})^2]= 1/2$.
	\section{Existence of a Nash equilibrium}
	\label{appendix:existence_nash_equilibrium}
	In this section, we set $h>0$ as a terminal time of the auction to investigate the game played by the market takers.
\subsection{Nash equilibrium}
\label{subsubsec:nash_equilibrium}
We are interested in finding a Nash equilibrium to the game between buyers and sellers. Starting at $(N_0^a,N^b_0)=(\alpha,\beta)\in \mathbb N^2$, we set\footnote{Rigorously speaking we should write $V_0^{i,\alpha,\beta}(\lambda_a,\lambda_b,h)$ instead of $V_h^{i,\alpha,\beta}(\lambda_a,\lambda_b)$ with $i\in \{a,b\}$, since we define here the value function of each market taker at time $0$ and $h$ is a time horizon. Since we consider only value functions of market takers at time $0$, we make this slight abuse of notation.}
\begin{equation}
\label{pb:buyer}
V^{a,\alpha,\beta}_h(\lambda_a,\lambda_b)=  \mathbb{E}^{\mathbb{P}^{\lambda_a, \lambda_b}}[N^a_h(N^a_h - N^b_h)]
\end{equation} 
\begin{equation}
\label{pb:seller}
V^{b,\alpha,\beta}_h (\lambda_a,\lambda_b)=  \mathbb{E}^{\mathbb{P}^{\lambda_a, \lambda_b}}[N^b_h(N^b_h - N^a_h)].
\end{equation}
Formally, we can thus compute the optimal P\&L of market takers for buy orders and sell orders by solving the following coupled system
	\begin{equation}
	\label{eq:market_taker_game}
	\left\{
		\begin{array}{lll}
				\underset{\lambda^a\in \mathcal U}\inf V^{a,\alpha,\beta}_h(\lambda_a,\lambda^\star_b)&=&  \mathbb{E}^{\mathbb{P}^{\lambda_a^\star, \lambda_b^\star}}[N^a_h(N^a_h - N^b_h)]\\
				 \underset{\lambda^b\in \mathcal U}\inf V^{b,\alpha,\beta}_h (\lambda_a^\star,\lambda_b)&=&  \mathbb{E}^{\mathbb{P}^{\lambda_a^\star, \lambda_b^\star}}[N^b_h(N^b_h - N^a_h)]\\
		\end{array}			
	 \right.,
	\end{equation}
where $\lambda^\star_b$ and $\lambda^\star_b$ are simultaneous optimizers of \eqref{pb:buyer} and \eqref{pb:seller} respectively (depending on the action of market takers having the opposite behavior).\\

We now investigate theoretically the existence of a Nash equilibrium associated with \eqref{eq:market_taker_game}. First we introduce some notations.

\begin{itemize}
\item Let $\Omega$ be the set of piece-wise constant functions with jumps of size $1$. Consider\footnote{Here for the notation $\top$ denotes the transposition of a vector to identify  as usual any element of $\mathbb N^2$ with a column vector.} $X=(N^a,N^b)^\top$ be the canonical processes in $\Omega^2$ and $\mathbb{F} = (\mathcal{F}_s)_{0\leq s \leq h}$ the smallest filtration for which $X$ is adapted. 
\item Let $\mathbb{P}$ be a probability measure on $(\Omega^2,\mathcal F_h)$ such that
$$
M_s = X_s - s\mathcal L_0, ~\text{with}~ \mathcal L_0:=(\lambda_0, \lambda_0)^\top,\; 0<\lambda_0<\lambda_{+},\; s\in [0,h],
$$ is a local martingale. A proof of the existence of such measure $\mathbb P$ is given in \cite{jacod1975multivariate}.
We set $M^a_r:=M_{1,r}$ (resp. $M^b_r:=M_{2,r}$) the first (resp. the second) component of $M$.
Moreover to any pair $(\lambda^a, \lambda^b)\in \mathcal{U}^2$ of admissible controls we associate $\mathbb{P}^{\lambda^a, \lambda^b}$ the measure defined by
$$
\frac{\mathrm{d}\mathbb{P}^{\lambda^a, \lambda^b}}{\mathrm{d}\mathbb P} = \text{exp}\Big(\int_0^h \log\big(\frac{\lambda_s^{a}}{\lambda_0}\big) \mathrm{d}N^a_s - \big(\lambda_s^{a} - \lambda_0\big)\mathrm{d}s + \log\big(\frac{\lambda_s^{b}}{\lambda_0}\big) \mathrm{d}N^b_s - \big(\lambda_s^{b} - \lambda_0\big) \mathrm{d}s \Big).
$$
Hence, under the measure $\mathbb{P}^{\lambda^a, \lambda^b}$, 
$$
\Big(X_s - \int_0^s (\lambda^a_u, \lambda^b_u)^\top\mathrm{d}u\Big)_{0\leq s\leq h}
$$
is a martingale. 

\item For $(E, \|\cdot\|)$ a normed space, any $0\leq s\leq t\leq h$ and $p>1$, we define
\[\mathcal{H}^p_{s,t}(E) = \{Y, ~E-\text{valued and }\mathbb F-\text{adapted process s.t.}, \mathbb{E}[(\int_s^t \|Y_r\|^2 \mathrm{d}r)^{\frac{p}{2}}]<+\infty\}\]
\[\mathcal{S}^p_{s,t}(E) = \{Y, ~E-\text{valued and }\mathbb F-\text{adapted process s.t.}, \mathbb{E}[\underset{s\leq t}{\sup} \|Y_r\|^p \mathrm{d}r]<+\infty\}\]
\[\mathbb{L}^p(E) = \{\xi,~E-\text{valued }\mathcal{F}_h-\text{measurable random variable, s.t.}~ \mathbb{E}[\|\xi\|^p] <+\infty \}.\] 

When $s=0$ we omit the index $s$ in the previous definitions. If $E=\mathbb R^2$, we set $\|\cdot\|_2$ and $\|\cdot\|_1$ the classical Manhattan norm and Euclidean norm on $\mathbb R^2$ respectively. For any $\mathbb R^2-$valued process $Y:=(Y_r)_{0\leq r\leq h}$, we denote by $Y_{r,1}$ and $Y_{r,2}$ its first and second coordinates respectively for any time $r\in [0,h]$. 
\item For any $z\in \mathbb R^2$ and $\varepsilon^a,\varepsilon^b\in[\lambda_-,\lambda_+]$, we set 

$$
\textbf{(L)}\left\{
\begin{array}{ll}
\lambda_a^{\star}(z, \varepsilon^a) &= \mathbf{1}_{z_1 > 0} \lambda_{-} + \mathbf{1}_{z_1 < 0} \lambda_{+} + \varepsilon^a \mathbf{1}_{z_1= 0}\\
\lambda_b^{\star}(z, \varepsilon^b) &= \mathbf{1}_{z_2 > 0} \lambda_{-} + \mathbf{1}_{z_2 < 0} \lambda_{+} + \varepsilon^b \mathbf{1}_{z_2= 0}.
\end{array}
\right.
$$ 

Note that both $z_1\lambda^{ \star}_a(z, \varepsilon^a)$ and $z_2\lambda^{ \star}_b(z, \varepsilon^b)$ do not depend on $\varepsilon^a$ and $\varepsilon^b$. To alleviate notations, when one of these products appears, we will denote it simply by $z_1\lambda^{ \star}_a(z)$ and $z_2\lambda^{ \star}_b(z)$ respectively.
\item For any $z,\tilde z\in \mathbb R^2$ and any $\varepsilon \in [\lambda_-,\lambda_+]$, we set $H^{a,\star}(z, \tilde z, \varepsilon)  = z_1\lambda_a^{\star}(z) + z_2\lambda_{b}^{\star}(\tilde z, \varepsilon)$ and $H^{b,\star}(z, \tilde z, \varepsilon)  = z_2\lambda_b^{\star}(z) + z_1\lambda_{a}^{\star}(\tilde z, \varepsilon)$.
\item for $x\in \mathbb{N}^2$ we define $g^a(x) = x_1(x_1 - x_2)$ and $g^b(x) = x_2(x_2 - x_1)$.
\item Let $U$ be a map from $[0,h]\times\mathbb{N}^2$ into $\mathbb R$. For any $(s,\alpha,\beta)\in [0,h]\times \mathbb{N}^2$ we set

$$
\mathbf{(D)}\left\{
\begin{array}{ll}
D_a U(s,\alpha, \beta) =& U(s,\alpha+1, \beta) - U(s,\alpha, \beta) \\
D_b U(s,\alpha, \beta) =& U(s,\alpha, \beta+1) - U(s,\alpha, \beta) \big)\\
DU(s,\alpha, \beta)=&(D_a U(s,\alpha, \beta),D_b U(s,\alpha, \beta))^\top.
\end{array}
\right. $$

\end{itemize}

We first provide a very general result by associated to the existence of a Nash equilibrium for \eqref{eq:market_taker_game} a system of coupled ODE on $\mathbb N^2$, as a direct extension of \cite[Theorem 8.5]{dockner2000differential}.

\begin{proposition}
\label{th:nash_equilibrium_bangbang}
Assume that there exist two maps $\varepsilon^a,\varepsilon^b$ from $[0,h]\times \mathbb{N}^2$ into $[\lambda_-,\lambda_+]$ such that the following coupled system 

$$
\mathbf{(S)}\left\{
\begin{array}{ll}
\partial_s V^a+H^{a,\star}(DV^a,DV^b,\varepsilon^b)=0,& s\in [0,h),\; (\alpha, \beta)\in \mathbb{N}^2\\
V^a(h,\alpha, \beta)=g^a(\alpha, \beta),&\\
\partial_s V^b+H ^{b,\star}(DV^b,DV^a,\varepsilon^a)=0,& s\in [0,h),\; (\alpha, \beta)\in \mathbb{N}^2\\
V^b(h,\alpha,\beta)=g^b(\alpha,\beta),&
\end{array}
\right.
$$
has a continuously differentiable (in time) solution denoted by $(V^a,V^b)$ on $[0,h]\times \mathbb{N}^2$ and assume moreover that 
$$DV^i(\cdot,N^a_\cdot,N^b_\cdot) \in \mathcal H_h^2(\mathbb R^2),\; i=a,b.$$ Then, $(\lambda^{\star}_a(D V^a,\varepsilon^a),\lambda^{\star}_b(D V^b,\varepsilon^b))$ is a Nash equilibrium for \eqref{eq:market_taker_game}.

\end{proposition}

\begin{proof}
The proof follows a standard verification argument. Notice however that we need feedback control for the thresholds $(\varepsilon^a,\varepsilon^b)$ in order to have classical HJB equations. See for instance \cite[Theorem 8.5]{dockner2000differential}.
\end{proof}
Although the previous result provides sufficient conditions to get a Nash equilibrium for the stochastic differential game \eqref{eq:market_taker_game}, it is quite hard to justify such existence in practice. Note indeed that the optimizers $\lambda^{\star}_a$ and $\lambda^{\star}_b$ are singular in view of their definition \textbf{(L)}. Thus, the main difficulty encountered in this proposition is to solve the bang-bang type system \textbf{(S)} of ODEs on $\mathbb{N}^2$ for relevant thresholds $\varepsilon^a,\varepsilon^b$. As far as we now, we have no PDE results ensuring the existence of a solution to \textbf{(S)}.\\

Inspired by \cite{hamadene2014bang}, we thus propose to study a smooth approximation of \textbf{(S)} and then to build a sequence of processes converging (up to a subsequence) to a Nash equilibrium for the game \eqref{eq:market_taker_game}.\\

Let $n\in \mathbb N$. We consider the smoothed control functions for any $z\in\mathbb R$
$$
\lambda^{n}(z) = \left\{ \begin{array}{lll}
\lambda_{+}~ &\text{if } z \leq -\frac{1}{n}\\
\lambda_{-}~ &\text{if } z \geq \frac{1}{n}\\
n\frac{\lambda_- - \lambda_+}{2}z + \frac{\lambda_+ + \lambda_-}{2}~ &\text{if }z \in (-\frac{1}{n}, \frac{1}{n}).
\end{array}\right.
$$
The functions $\lambda^{n}$ and $z\longmapsto z\lambda^{n}(z)$ are Lipschitz continuous. Also consider $\Phi_n$, the truncation function defined for any $x\in \mathbb R$ by
$$
\Phi_n(x) = (x \wedge n)\vee (-n).
$$

Hence, we introduction the smoother of $H^\star$ denoted by $H^{\star,n}$ and defined by for any $(z_1,z_2,\tilde z)\in \mathbb R^3$ by

$$ H^{\star,n}(z_1,z_2,\tilde z)= \Phi_n(z_1\lambda^{n}(z_1)) + \Phi_n(z_2)\lambda^{n}(\tilde z).$$

\begin{theorem}\label{thm:edpapproach}
For any $n\in \mathbb N$, there exists a unique (viscosity) solution denoted by $V^{a,n}$ to the following system of integro-PDEs

$$
\mathbf{(S^n)}\left\{
\begin{array}{ll}
\partial_s V^{a,n}+H^{\star,n}(D_aV^{a,n},D_bV^{a,n}, D_bV^{b,n})=0,\; s\in [0,h),\, (\alpha,\beta)\in \mathbb{N}^2,&\\
V^{a,n}(h,\alpha, \beta)=g^a(\alpha, \beta),&\\
\partial_s V^{b,n}+H ^{\star,n}(D_bV^{b,n},D_aV^{b,n}, D_aV^{a,n})=0,\,  s\in [0,h),\, (\alpha,\beta)\in \mathbb{N}^2,&\\
V^{b,n}(h,\alpha,\beta)=g^b(\alpha, \beta).&
\end{array}
\right.
$$
Moreover,
\begin{itemize}
\item The system $(\mathbf{S^n})$ admits a unique viscosity solution.
\item There exists a subsequence $(n_k)_{k\geq 0}$ and two measurable applications $V^a,V^b$ from $[0,h]\times \mathbb{N}^2$ into $\mathbb R$ such that for any $(s,\alpha,\beta)\in [0,h]\times \mathbb{N}^2$ 
$$\lim\limits_{k\to+\infty} V^{i,n_k}(s,\alpha,\beta)=V^i(s,\alpha,\beta), \; i\in\{ a,b\}$$
and 
$$\lim\limits_{n\to+\infty} DV^{i,n}(s,\alpha,\beta)=DV^i(s,\alpha,\beta),\; i\in\{a,b\}.$$
\item Moreover $\lambda^{n_k}(D_aV^{a,n_k}(\cdot,N^a,N^b))\mathbf 1_{D_aV^{a}(\cdot,N^a,N^b)=0}$ and $\lambda^{n_k}(D_bV^{b,n_k}(\cdot,N^a,N^b))\mathbf 1_{D_bV^{b}(\cdot ,N^a,N^b)=0}$ converges weakly in $\mathcal H_h^2(\mathbb R^2)$ to some progressively measurable and $[\lambda_-,\lambda^+]$-valued processes denoted respectively by $\theta$ and $\vartheta$. 
\end{itemize}
Thus, $(\lambda^{\star}_a,\lambda^\star_b)=(\lambda^{\star}_a(DV^a(s,N_s^a,N_s^b),\theta_s),\lambda^{\star}_b(DV^b(s,N_s^a,N_s^b),\vartheta_s))_{0\leq s\leq t}$ is a Nash equilibrium for the game \eqref{eq:market_taker_game} and $V_h^{i, \alpha, \beta}(\lambda^{\star}_a,\lambda^{\star}_b)=V^i(0,\alpha,\beta),\; i\in\{a,b\}$.
\end{theorem}

We give here the sketch of the proof of this result. The details are postponed to Appendix \ref{appendix_sec:proof_theorem}. 

\paragraph{Sketch of the proof of Theorem \ref{thm:edpapproach}}
The proof will be divided in three steps. The main tool used is the theory of BSDE with jumps (see \cite{tang1994necessary,buckdahn1994bsde,barles1997backward}) and their representations through integro-partial differential equations.

\begin{itemize}
\item[Step 1.] We associated to the system $\mathbf{(S^n)}$ a two dimensional BSDE for which it is well-known that there exists a unique solutions in appropriate spaces.
\item[Step 2.] By mimicking the proof of Theorem 2.5 in \cite{hamadene2014bang} extended to the case of counting processes, we prove that the solution of the BSDE associated to $\mathbf{(S^n)}$ converges up to a subsequence to a solution of a two-dimensional BSDE associated with the system $\mathbf{(S)}$.
\item[Step 3.] We prove that this approximation provides a Nash equilibrium for the game \eqref{eq:market_taker_game} with well-chosen thresholds obtained in Step 2 as limits of functions of the components of the solution to the approached BSDE considered, see Proposition \ref{th:nash_equilibrium} below.\\

We conclude thanks to semi-linear Feynman-Kac formula for BSDEs and the system $\mathbf{(S^n)}$ established in Step 1, together with convergence results.
\end{itemize}

\subsection{Proof of Theorem \ref{thm:edpapproach}}
\label{appendix_sec:proof_theorem}
For the proof we follow the methodology of  \cite{hamadene2014bang}. First we introduce a series of smoothed BSDE with Lipschitz generator by smoothing the controls $\lambda^{\star}_a, ~\lambda^{\star}_b$. Then we show that the solution of the smoothed BSDE converges (up to a subsequence) almost surely towards a solution of Equation \eqref{eq:bsde}.\\ 

\noindent We have the following a priori estimates results which is a consequence of the BDG inequalities and of the Gronwall Lemma.
\begin{lemma}
\label{lemma:apriori_estimates}
For $(s, x)\in [0,h] \times \mathbb{N}^2 $ let $X^{s, x}$ be the process in $\Omega$ defined onto $[s, h]$ by
$$
X^{s, x}_u = x + X_{u} - X_s.
$$
We have for any $s\in [0,h]$ and $\rho>0$
$$
\mathbb{E}[\underset{s\leq u\leq h}{\sup}\|X^{s, x}_u \|^{\rho}_1] \leq C_{\rho}(1 + |x_1|^{\rho} + |x_2|^{\rho})
$$
and for any $(\lambda_a, \lambda_b)\in \mathcal{U}^2$
$$
\mathbb{E}^{\mathbb{P}^{\lambda_a, \lambda_b}}[\underset{s\leq u\leq h}{\sup}\|X^{s, x}_u \|^{\rho}_1] \leq C_{\rho}(1 + |x_1|^{\rho} + |x_2|^{\rho})
$$
\end{lemma}
We now turn to the proof of Theorem \ref{thm:edpapproach}. 
\subsubsection{Step 1: Approximation, existence and uniqueness}
\label{appendix_subsec:step_1}
From now, $s\in [0,h)$. We recall the definition of smoothed control functions
$$
\lambda^{n}(z) = \left\{ \begin{array}{lll}
\lambda_{+}~ &\text{if } z \leq -\frac{1}{n}\\
\lambda_{-}~ &\text{if } z \geq \frac{1}{n}\\
n\frac{\lambda_- - \lambda_+}{2}z + \frac{\lambda_+ + \lambda_-}{2}~ &\text{if }z \in (-\frac{1}{n}, \frac{1}{n})
\end{array}\right.
$$
Consider $\Phi_n$, the truncation function
$$
\Phi_n(x) = (x \wedge n)\vee (-n).
$$
Now we define the system of smoothed BSDEs for any $u\in [s,h]$:
$$\mathbf{(J^n)}
\left\{ \begin{array}{rl}
- \mathrm{d}Y^{a, n;s,x}_u& = (H^{\star,n}(Z^{a, n;s, x}_{1, u},Z^{a, n;s, x}_{2, u},Z^{b, n;s, x}_{2,u})-\mathcal L_0\cdot Z^{a, n;s, x}_u) \mathrm{d}u - Z^{a, n;s, x}_u\cdot \mathrm{d}M_u,\\
 Y^{a, n;s, x}_{h}& = g^a(X_h^{s,x}) \\
- \mathrm{d}Y^{b, n;s,x}_u& = (H^{\star,n}(Z^{b, n;s, x}_{2, u},Z^{b, n;s, x}_{1, u},Z^{a, n;s, x}_{1,u}) -\mathcal L_0\cdot Z^{b, n;s, x}_u) \mathrm{d}u - Z^{b, n;s, x}_u\cdot \mathrm{d}M_u,\\
 Y^{b, n;s, x}_{h}& = g^b(X_h^{s,x}),
\end{array}\right.
$$
with $Z^{i, n;s, x}_{u}=(Z^{i, n;s, x}_{1, u},Z^{i, n;s, x}_{2, u})^\top$ for any $i\in \{a,b\}$.\\

From Proposition 2.1. in \cite{buckdahn1994bsde} since $\Phi_n$ is Lipschitz continuous there exists a unique solution to $\mathbf{(J^n)}$ such that
$$\big((Y^{a, n;s, x}, Z^{a, n;s, x}), (Y^{b, n;s, x}, Z^{b, n;s, x})\big)\in \big(\mathcal{S}^2_{s,h}(\mathbb{R})\times \mathcal{H}^2_{s,h}(\mathbb{R}^2)\big)^2.$$ Moreover (Proposition 3.8. in \cite{buckdahn1994bsde}) there exist measurable deterministic functions $V^{a,n},~ V^{b,n}$  defined on $[s, h]\times \mathbb{N}^2$ with values in $\mathbb{R}$ such that:
\begin{equation}\label{feynmankac}
\forall u \in [s, h], ~ Y^{i, n;s, x}_u = V^{i, n}(s, X^{s, x}_u)~ \text{and} ~ Z^{i, n;s, x}_u = DV^{i, n}(u, X^{s, x}_{u^{-}}), ~ \text{for}~ i = a, b.
\end{equation}

From Theorem 3.4. in \cite{barles1997backward}, we know that the unique solution of $\mathbf{(J^{n})}$ provides a unique viscosity solution denoted by $(V^{a,n},V^{b,n})$ to $\mathbf{(S^n)}$ and given by \eqref{feynmankac}.\\

Before going to the convergence of $Y^{i,n}$ and $Z^{i,n}$, notice that by considering the generator functions
$$
\left\{ \begin{array}{ll}
H^{a, n}(u, x) &= \big( \Phi_n(D_aV^{a, n}(u, x)\lambda^{n}(D_aV^{a, n}(u, x))) + \Phi_n(D_bV^{a, n}(u, x))\lambda^{n}(D_bV^{b, n}(u, x)) \big)\\
H^{b, n}(u, x) &= \big( \Phi_n(D_b V^{b, n}(u, x)\lambda^{n}(D_bV^{b, n}(u, x))) + \Phi_n(D_aV^{b, n}(u, x))\lambda^{n}(D_aV^{a, n}(u, x)) \big),
\end{array}\right.
$$
we deduce from \eqref{feynmankac} that \[H^{a, n}(u, X^{s, x}_{u^-})=H^{\star,n}(Z^{a, n;s, x}_{1, u},Z^{a, n;s, x}_{2, u},Z^{b, n;s, x}_{2,u}),\]
and \[H^{b, n}(u, X^{s, x}_{u^-})= H^{\star,n}(Z^{b, n;s, x}_{2, u},Z^{b, n;s, x}_{1, u},Z^{a, n;s, x}_{1,u}), \]
so that $\mathbf{(J^n)}$ becomes
$$\mathbf{(\widetilde{J^{n}})}
\left\{ \begin{array}{ll}
- \mathrm{d}Y^{a, n;s,x}_u &= (H^{a, n}(s, X^{s, x}_{u^-})-\mathcal L_0\cdot Z^{a, n;s, x}_u) \mathrm{d}u - Z^{a, n;s, x}_u\cdot \mathrm{d}M_u, ~ Y^{a, n;s, x}_{h} = g^a(X^{s,x}_h) \\
- \mathrm{d}Y^{b, n;s,x}_u &= (H^{b, n}(u, X^{s, x}_{u^-})-\mathcal L_0\cdot Z^{b, n;s, x}_u)\mathrm{d}u - Z^{b, n;s, x}_u\cdot \mathrm{d}M_u, ~ Y^{b, n;s, x}_{h} = g^b(X^{s,x}_h).
\end{array}\right.
$$

\subsubsection{Step 2: Convergence to the solution of a bang-bang system of BSDEs}\label{appendix_subsec:step_2}
From now, we consider any index $i$ equals to $a$ or $b$, we set $x\in \mathbb N^2$ and $s\in [0,h]$.\\

\textbf{Step 2a. Uniform estimates.}\vspace{0.5em}

In order to use dominated convergence we give some uniform a priori estimates for processes $(Y^{i, n;s, x}, Z^{i, n;s, x}).$\\

We first aim at using a comparison principle to control the upper bound of $Y^{i,n}$ and introduce the following BSDE

\begin{equation}\label{BSDE+}
\overline{Y}^{i, n;s,x}_u = g^i(X^{s, x}_h) + \int_u^h 4\lambda^{+}\| \overline{Z}^{i, n;s,x}_r\|_1 \mathrm{d}r - \int_u^h \overline{Z}^{i, n;s,x}_r \cdot \mathrm{d}M_r,\; s\leq u\leq h.
\end{equation}

Once again according to \cite{buckdahn1994bsde} there exists a unique solution $(\overline{Y}^{i, n;s,x}, \overline{Z}^{i, n;s, x})$ of the above BSDE in the space $\mathcal{S}^2_{s,h}(\mathbb{R})\times \mathcal{H}^2_{s,h}(\mathbb{R}^2)$ and there exists deterministic measurable functions $\overline{V}^{i, n}$ such that for any $u\in [s, h]$:
$$
\overline{Y}^{i, n;s,x}_u = \overline{V}^{i, n}(u, X_u^{ s,x}).
$$ 
 By comparison theorem for BSDE (see for instance\footnote{To be more accurate, we identify  our pair of processes as a compound Poisson process with jumps in $\{-1,1\}$, so that we are in the framework of \cite{royer2006backward} for a compensator $\lambda(dx)=\lambda_0(\delta_1(dx)+\delta_{-1}(dx)).$} Theorem 2.5 in \cite{royer2006backward}), for any time $s\leq u\leq h$ we get
\begin{equation}\label{comparisonBSDE}
Y^{i, n; s,x}_u\leq \overline{Y}_u^{i, n;s,x}, \mathbb{P}-a.s.
\end{equation}
We now give a uniform estimates of $\overline{Y}^{i,n;s,x}$ to get a uniform estimates for $Y^{i, n; s,x}$ in view of the previous relation. Consider the bi-dimensional process:
$$
M^{i, n}_{u}= M_u - 4\lambda_{+} \text{sign} (\overline{Z}^{i, n;s, x}_u),
$$
where the sign is taken coordinate by coordinate. The process $M^{i, n} = (M^{i, n}_1, M^{i, n}_2)$ is a bi-dimensional martingale under the probability $\mathbb{P}^{i, n}$ equivalent to $\mathbb{P}$ with density given by
$$
\mathcal E_h^{i, n} = \text{exp}\big( \int_0^h \log(\frac{\gamma^{i, n}_{t, 1}}{\lambda_0})\mathrm{d}N^{a}_t - (\gamma^{i, n}_{t, 1}- \lambda_0)\mathrm{d}t + \log(\frac{\gamma^{i, n}_{t, 2}}{\lambda_0})\mathrm{d}N^{b}_t - (\gamma^{i, n}_{t, 2} - \lambda_0)\mathrm{d}t  \big)
$$
with
$$
\gamma^{i, n}_{t, j} = \lambda_0 + 4\lambda_{+} \text{sign} (\overline{Z}^{i, n;t, x}_{j,t}).
$$
Consequently we get
$$
\overline{V}^{i, n}(s, x) = \mathbb{E}^{\mathbb{P}^{i, n}}[g_i(X_h^{s, x})].
$$
By polynomial growth of $g_i$ we deduce that there exists a positive constant $\tilde C$ such that
$$
|\overline{V}^{i, n}(s, x)| \leq \tilde  C \mathbb{E}^{\mathbb{P}^{i, n}}[\| X^{s, x}_h\|_2^2].
$$
Note that there exists a positive constant $\tilde \kappa$ such that

$$
\mathbb{E}^{\mathbb{P}^{i, n}}[\|X^{s, x}_h\|^2_2] \leq  \tilde \kappa(\|x\|_2^2 + 1).
$$
The previous equation implies the following polynomial growth bound
$$
|\overline{V}^{i, n}(s, x)|\leq C(1 + \|x\|^2_2),
$$
where $C:=\tilde C\tilde \kappa >0$.

According to the comparison result \eqref{comparisonBSDE} together with \eqref{feynmankac}, we deduce that there exists some positive constant $C$, which does not depend on $n$, such that

\[ 
V^{i,n}(s,x) \leq C(1 + |x_1|^{2} + |x_2|^{2}).
\]

Similarly, by considering a BSDE similar to \eqref{BSDE+} but with a minus sign in the generator, we get 
\[ 
V^{i,n}(s,x) \geq -C(1 + |x_1|^{2} + |x_2|^{2}).
\]

We thus deduce that for any $(s, x)\in [0,h]\times \mathbb{N}^2$ and $p\geq 1$ the following estimate holds for some positive constant $C_p$
\begin{equation}
\label{eq:proof_control_inequality_a}
\mathbb{E}[~\underset{s\leq u\leq h}{\sup}~|Y^{i, n;s,x}_u|^{p}]\leq C_{p}(1 + |x_1|^{2p} + |x_2|^{2p}).
\end{equation}

Moreover, the characterization \eqref{feynmankac} allows to transfer the prior estimates of $Y^{i, n;s, x}$ to $Z^{i, n;s, x}$. In particular we get that for any $p\geq 1$
\begin{equation}
\label{eq:proof_control_inequality_b}
\mathbb{E}[\underset{s\leq u \leq h}{\sup} |Z^{i, n;s, x}_u|^p ] \leq C_p (1 + |x_1|^{2p}+ |x_2|^{2p}).
\end{equation}
Note that the constant $C_p$ does not depend on $n$, so that Estimates \eqref{eq:proof_control_inequality_a} and \eqref{eq:proof_control_inequality_b} are uniform with respect to $n$.\\

\textbf{Step 2b. Convergence of the solutions of the smoothed BSDE.}\\

We now turn to the convergence of $(Y^{i, n;s, x}, Z^{i, n;s, x}),$ in $\mathcal{S}^2_{s,h}(\mathbb{R})\times \mathcal{H}^2_{s,h}(\mathbb{R}^2)$. For any $q\leq 2$, there exists a positive constant $\tilde C$ which does not depend on $n$ such that
\[
\mathbb{E}[\int_0^h|H^{i, n}(r, X_{r^-}^{0, 0})|^{q}\mathrm{d}r \leq \mathbb{E}[\int_0^h 2\lambda_{+}\|Z^{i, n;0, 0}_r\|_1^q \mathrm{d}r]\leq \tilde C.
\]
The sequence $(H^{i, n})_{n\geq 0}$ is bounded in $\mathbb{L}^2\big([0,h]\times\mathbb{N}^2, \mathrm{d}r \times \mu(0, 0;r,\mathrm{d}x) \big)$ where $\mu(0, 0;r,\mathrm{d}x)$ is the law of $X^{0, 0}_{r^-}$ under $\mathbb{P}$. Thus there exists a subsequence $(n_k)_{k\geq 0}$ such that $(H^{i, n_k})_{k\geq  0}$ converges weakly in $\mathbb{L}^2([0,h]\times\mathbb{R}, \mu(0, 0;r,\mathrm{d}x)\mathrm{d}r)$. We omit the index $k$ and still write $n$ instead of $n_k$ to reduce the notations.\\

We now prove that for any $(s, x)\in [0, h]\times \mathbb{N}^2$, $(V^{i, n}(s, x))_{n\geq0}$ is a Cauchy sequence. We set the function $\Delta^{i,n,m}(t,x,z_n,z_m):=H^{i, n}(t,x)- H^{i, m}(t,x)-\mathcal L_0\cdot(z_n-z_m)$ with $(n,m)\in \mathbb N$ and $(t,x,z_n,z_m)\in [0,T]\times \mathbb N^2\times \mathbb R^2\times \mathbb R^2$. Let $\delta\in [0,h-s]$ and $k\in \mathbb N$, we have
\begin{align}
\nonumber |V^{i, n}(s, x) - V^{i, m}(s, x)| &= |\mathbb{E}[\int_s^h {\Delta^{i,n,m}}(r, X^{s, x}_{r^-}, Z^{i, n;s, x}_r, Z^{i, m;s, x}_r)\mathrm{d}r]|\\
\label{VnVm} &\leq E_-^{s+\delta,h}+ E_+^{s+\delta,h} +E^{s,s+\delta}, \end{align}
with
\[  E_-^{s+\delta,h}:= |\mathbb{E}[\int_{s+\delta}^h \mathbf{1}_{\|X^{s, x}_{r^-}\|_{\infty}\leq k}{\Delta^{i,n,m}}(r, X^{s, x}_{r^-}, Z^{i, n;s, x}_r, Z^{i, m;s, x}_r)\mathrm{d}r]|,\]
\[E_+^{s+\delta,h}:=|\mathbb{E}[ \int_{s+ \delta}^h \mathbf{1}_{\|X^{s, x}_{r^-}\|_{\infty}> k}{\Delta^{i,n,m}}(r, X^{s, x}_{r^-}, Z^{i, n;s, x}_r, Z^{i, m;s, x}_r) \mathrm{d}r]|,\]
and \[E^{s,s+\delta}:= |\mathbb{E}[\int_{s}^{s+\delta} {\Delta^{i,n,m}}(r, X^{s, x}_{r^-}, Z^{i, n;s, x}_r, Z^{i, m;s, x}_r)\mathrm{d}r]|.\]
We obtain from \eqref{eq:proof_control_inequality_b} that there exists some constant $C$ independent of $n$ and $m$ such that
\begin{align*}
E^{s,s+\delta}\leq  C\delta.
\end{align*}

We now turn to $E^{s+\delta,h}_+$. By using Cauchy Schwarz and Markov inequalities together with the prior inequalities \eqref{eq:proof_control_inequality_a} and \eqref{eq:proof_control_inequality_b}, there exists a positive constant $\hat C$ again independent of $n$ and $m$ such that for any positive integer $k$
\begin{align*}
E_+^{s+\delta,h}& \leq |\mathbb{E}[ \int_{s+ \delta}^h \mathbf{1}_{\|X^{s, x}_{r^-}\|_{\infty}> k} \mathrm{d}r]|^{\frac{1}{2}} |\mathbb{E}[ \int_{s+ \delta}^h \Delta^{i,n,m}(r, X^{s, x}_{r^-}, Z^{i, n;s, x}_r, Z^{i, m;s, x}_r)^2 \mathrm{d}r]|^{\frac{1}{2}}\\
&\leq \frac{\hat C}{\sqrt{k}}.
\end{align*}
Finally, we note that
\begin{align*}
E_-^{s+\delta,h}&= \big|  \sum_{(p, q) \in \mathbb{N}^2} \int_{s}^h  \Delta^{i,n,m}(r,p,q, D V^{i, n}(t,p,q), D V^{i,m}(t,p,q)) \mathbb{P}\big(X^{t,(0, 0)}_{r} = (p,q) \big)  \phi_{s, x}(r, p, q) \mathrm{d}r \big| 
\end{align*}
with
$$
\phi_{s, x}(r, p, q) = \mathbf{1}_{p\leq k}\mathbf{1}_{q\leq k}\mathbf{1}_{r\geq s + \delta}  \frac{\mathbb{P}\big(X^{t,x}_{r} = (p,q) \big) }{\mathbb{P}\big(X^{t,(0, 0)}_{r} = (p,q) \big) }.
$$
Since 
$$
\mathbb{P}\big(X^{t,(0, 0)}_{r} = (p,q) \big)^{-1} = e^{2\lambda_0 r}\frac{p ! q!}{(\lambda_0r)^{p + q}}
$$
is bounded for $p$ and $q$ lower than $k$ and $r$ lower than $h$. The function $\phi_{s, x}$ is bounded and thus in $\mathbb{L}^{2}([0, h]\times  \mathbb{N}^2, \mu(0,0;s, \mathrm{d}x)\times \mathrm{d}s))$ consequently by weak convergence of $H^{i, n}$, we have that $E_-^{s+\delta,h}$ goes to $0$ when $m,n$ go to infinity. Hence, taking the limit when $\delta$ goes to $0$ and $k,n,m$ go to infinity, we deduce from \eqref{VnVm} that $(V^{i, n}(s, x))_{n\geq 0 }$ is a Cauchy sequence. We thus denote by $V^i(s,x)$ the limit of $(V^{i, n}(s, x))_{n\geq 0 }$. We recall that $V^i$ depends on the subsequence $(n_k)_{k\geq 0}$\\

We have the $\mathbb{P}$-almost sure convergence (up to the subsequence) of $Y_u^{i, n; s, x}$ since
$
Y_u^{i, n; s, x} = V^{i, n}(u, X^{s, x}_u).
$ We denote by $Y^{i;s, x}$ the almost sure limit of $Y^{i, n; s, x}$.
Notice moreover that in view of $\mathbf{(D)}$, we have\begin{equation}\label{cv:pointwise}
\lim\limits_{n\to+\infty} DV^{i,n}(s,x)=DV^i(s,x),\; (s,x)\in [0,h]\times \mathbb N^2.
\end{equation} 

By Equation \eqref{eq:proof_control_inequality_a} and Lebesgue dominated convergence theorem we have for any $\rho\geq 1$
\begin{equation}\label{conv:Y:Hrho}
\mathbb{E}[\int_{s}^h |Y_r^{i, n;s, x} - Y_r^{i;s, x}|^{\rho}\mathrm{d}r]\underset{n\rightarrow + \infty}{\rightarrow} 0.
\end{equation}
Let now $n,m$ be two positive integers. From Ito's formula applied to $(Y^{i, n;s, x} -Y^{i, m;s, x} )^2$ we get for any $s\leq u \leq h$ 

\begin{align}
&\nonumber |Y^{i, n;s, x}_u - Y^{i, m;s, x}_u|^2\\
&\nonumber = -\int_u^h |Z^{i, n;s, x}_{1,r} - Z^{i, m;s, x}_{1,r}|^2 \mathrm{d}(M_r^a + \lambda_0 r)-\int_u^h |Z^{i, n;s, x}_{2,r} - Z^{i, m;s, x}_{2,r}|^2 \mathrm{d}(M_r^b + \lambda_0 r)\\
\nonumber & \quad + 2\int_u^h(Y^{i, n;s, x}_r - Y^{i, m;s, x}_r)\big((H^{i, n} - H^{i, m})(r, X^{t, x}_{r^-})- \mathcal L_0\cdot (Z^{i, n;s, x}_r-Z^{i, m;s, x}_r)\big) \mathrm{d}r\\
\label{ineg:Yisn} &\quad -2 \int_u^h(Y^{i, n;s, x}_r - Y^{i, m;s, x}_r)(Z^{i, n;s, x}_r-Z^{i, m;s, x}_r)\cdot \mathrm{d}M_r.
\end{align}
Using Young's inequality and the definitions of $H^n$ and $H^m$ we deduce that there exists a positive constant $\tilde c$ (independent of $n$ and $m$) such that for any $\varepsilon>0$ 
\begin{align*}
&|Y^{i, n;s, x}_u - Y^{i, m;s, x}_u|^{2} + \int_u^h \lambda_0 \|Z^{i, n;s, x}_r - Z^{i, m;s, x}_r\|_2^2 \mathrm{d} r  \\
&\leq  \tilde c\varepsilon  |\lambda_+|^2 \int_u^h \big(\|Z^{i, n;s, x}_r\|_2^2 + \|Z^{i, m;s, x}_r\|_2^2 \big)\mathrm{d}r+ \frac{1}{\varepsilon} \int_u^h |Y^{i, n;s, x}_r - Y^{i, m;s, x}_r|^{2} \mathrm{d}r \\
&\quad  -2 \int_u^h(Y^{i, n;s, x}_r - Y^{i, m;s, x}_r)(Z^{i, n;s, x}_r-Z^{i, m;s, x}_r)\cdot \mathrm{d}M_r\\
&\quad -\int_u^h |Z^{i, n;s, x}_{1,r} - Z^{i, m;s, x}_{1,r}|^2 \mathrm{d}M_r^a -\int_u^h |Z^{i, n;s, x}_{2,r} - Z^{i, m;s, x}_{2,r}|^2 \mathrm{d}M_r^b.
\end{align*}
For $u = s$, by taking the expectation and by choosing $n,m$ large enough, we obtain from \eqref{eq:proof_control_inequality_b} and \eqref{cv:pointwise}, \eqref{conv:Y:Hrho} and the fact that $\varepsilon$ is arbitrary small that the following convergence holds 
\begin{equation}\label{conv:Z:cauchy}
\underset{n, m\rightarrow+\infty}{\lim \text{sup} }\mathbb{E}[\int_s^h \|Z^{i, n;s, x}_r - Z^{i, m;s, x}_r\|_2^{2} \mathrm{d}r]=0.
\end{equation}
Hence, $(Z^{i, n;s,x})_{n\in \mathbb{N}}$ is a Cauchy sequence (along the subsequence) and thus converges in $\mathcal{H}^2_{s,h}(\mathbb{R}^2)$ to some process $(Z_u^{i;s,x})_{s\leq u\leq h}$.\\

Similarly, by using \eqref{ineg:Yisn} and by noting that $-\int_u^h |Z^{i, n;s, x}_{1,r} - Z^{i, m;s, x}_{1,r}|^2 \mathrm{d}(M_r^a + \lambda_0 r)\leq 0$ and $-\int_u^h |Z^{i, n;s, x}_{2,r} - Z^{i, m;s, x}_{2,r}|^2 \mathrm{d}(M_r^b + \lambda_0 r)\leq 0$ since $M_\cdot^\alpha+\lambda_0\cdot =X_{\cdot}$ is a non decreasing process for $\alpha\in \{a,b\}$, we have
\begin{align*}
&\mathbb{E}[\underset{u\in[s, h]}{\sup}|Y^{i, n;s, x}_u - Y^{i, m;s, x}_u|^2 ]\\
 &\leq   \tilde c\varepsilon  |\lambda_+|^2 \mathbb E[\int_s^h \big(\|Z^{i, n;s, x}_r\|_2^2 + \|Z^{i, m;s, x}_r\|_2^2 \big)\mathrm{d}r]+ \frac{1}{\varepsilon} \mathbb E[\int_s^h |Y^{i, n;s, x}_r - Y^{i, m;s, x}_r|^{2} \mathrm{d}r]\\
 &+2 \mathbb E[\int_0^h|Y^{i, n;s, x}_r - Y^{i, m;s, x}_r||Z^{i, n;s, x}_{1,r}-Z^{i, m;s, x}_{1,r}| (dN_r^a+\lambda_0dr)]\\
 &+2 \mathbb E[\int_0^h|Y^{i, n;s, x}_r - Y^{i, m;s, x}_r||Z^{i, n;s, x}_{2,r}-Z^{i, m;s, x}_{2,r}| (dN_r^b+\lambda_0dr)]\\
 &\leq \tilde c\varepsilon  |\lambda_+|^2 \mathbb E[\int_s^h \big(\|Z^{i, n;s, x}_r\|_2^2 + \|Z^{i, m;s, x}_r\|_2^2 \big)\mathrm{d}r]+ \frac{1}{\varepsilon} \mathbb E[\int_s^h |Y^{i, n;s, x}_r - Y^{i, m;s, x}_r|^{2} \mathrm{d}r]\\
 &+2 \mathbb E[\int_0^h|Y^{i, n;s, x}_r - Y^{i, m;s, x}_r||Z^{i, n;s, x}_{1,r}-Z^{i, m;s, x}_{1,r}| (dM_r^a+2\lambda_0dr)]\\
 &+2 \mathbb E[\int_0^h|Y^{i, n;s, x}_r - Y^{i, m;s, x}_r||Z^{i, n;s, x}_{2,r}-Z^{i, m;s, x}_{2,r}| (dM_r^b+2\lambda_0dr)]\\
  &\leq \tilde c\varepsilon  |\lambda_+|^2 \mathbb E[\int_s^h \big(\|Z^{i, n;s, x}_r\|_2^2 + \|Z^{i, m;s, x}_r\|_2^2 \big)\mathrm{d}r]+ \frac{1}{\varepsilon} \mathbb E[\int_s^h |Y^{i, n;s, x}_r - Y^{i, m;s, x}_r|^{2} \mathrm{d}r]\\
 &+4\lambda_0 \mathbb E[\int_0^h|Y^{i, n;s, x}_r - Y^{i, m;s, x}_r|\|Z^{i, n;s, x}_{r}-Z^{i, m;s, x}_{r}\|_1 dr].
\end{align*}
By using again Young inequality for the last term in the previous inequality with the same $\varepsilon$, we deduce that there exists a positive constant $c>0$ independent of $n,m$ and $\varepsilon$ such that 
\begin{align*}
&\mathbb{E}[\underset{u\in[s, h]}{\sup}|Y^{i, n;s, x}_u - Y^{i, m;s, x}_u|^2 ]\\
&\leq c\big(\varepsilon  |\lambda_+|^2 \mathbb E[\int_s^h \big(\|Z^{i, n;s, x}_r\|_2^2 + \|Z^{i, m;s, x}_r\|_2^2 \big)\mathrm{d}r]+ \frac{1}{\varepsilon} \mathbb E[\int_s^h |Y^{i, n;s, x}_r - Y^{i, m;s, x}_r|^{2} \mathrm{d}r]\big).
\end{align*}
Since $\varepsilon$ is arbitrary and because of Equations \eqref{eq:proof_control_inequality_a}, \eqref{eq:proof_control_inequality_b} and \eqref{conv:Y:Hrho} we deduce
$$
\underset{n, m\rightarrow+\infty}{\lim \text{sup} }\mathbb{E}[\underset{u\in[s, t]}{\sup}|Y^{i, n;s, x}_u - Y^{i, m;s, x}_u|^2] =0.
$$
So we have the convergence of $(Y^{i, n;s, x})_{n\geq 0}$ in $\mathcal{S}^2_{s,h}(\mathbb{R})$ towards a process $(Y_u^{i;s,x})_{s\leq u\leq h}$ up to a subsequence.\\

\textbf{Step 2c. Convergence of the generator}\\

We study the convergence of $(H^{i, n})_{n\geq 0 }$, for $i \in \{ a, b\}$ (still along the subsequence introduced in Step 2b.). We focus on $(H^{a, n})_{n\geq 0 }$, the proof is identical for $(H^{b, n})_{n\geq 0 }$. Recall that
$$
H^{a, n}(u, X^{s, x}_{u^-}) = \Phi_n(Z^{a, n;s, x}_{1,u}\lambda^{\star}_a(Z^{a, n;s, x}_{1,u})) + \Phi_n(Z^{a, n;s, x}_{2,u})\lambda^{n}(Z^{b, n;s, x}_{2,u}).
$$
First note that
$$
\Phi_n(Z^{a, n;s, x}_{ 1, u}\lambda^{\star}_a(Z^{a, n;s, x}_{1,u})) \underset{n \rightarrow  +\infty}{\rightarrow } Z^{a;s, x}_{1,u}\lambda^{ \star}_a(Z^{a;s, x}_{1,u})
$$
with convergence taking place $\mathbb{P}$-a.s. and in $\mathcal{H}_{s,u}^2(\mathbb{R}^2)$ by dominated convergence and uniform integrability of $(\|Z^{a, n;s, x} \|_2^2)_{n\geq 0}$. We split the remaining part in a continuous and a non continuous parts
$$
\Phi_n(Z^{a, n;s, x}_{2,u})\lambda^{n}(Z^{b, n;s, x}_{2,u}) = \Phi_n(Z^{a, n;s, x}_{2,u})\lambda^{n}(Z^{b, n;s, x}_{2,u})\mathbf{1}_{Z^{b;s, x}_{u}\neq 0} + \Phi_n(Z^{a, n;s, x}_{2,u})\lambda^{n}(Z^{b,n;s, x}_{2,u})\mathbf{1}_{Z^{b;s, x}_{2,u} =  0}.
$$
We have the convergence of 
$\Phi_n(Z^{a, n;s, x}_{2,u})\lambda^{n}(Z^{b, n;s, x}_{2,u})\mathbf{1}_{Z^{b;s, x}_{2,u}\neq 0},\; \mathrm{d}s\times\mathrm{d}\mathbb P~a.e$
and the convergence also holds in $\mathcal{H}_{s,h}^2(\mathbb{R}^2)$. Moreover, $(\lambda^n(Z^{b, n;s, x}_{2,u})\mathbf{1}_{ Z^{b;s, x}_{u} = 0} )_{n\geq 0}$ being bounded we denote by $\vartheta$ a weak limit in $\mathcal{H}^2_{s,h}(\mathbb{R}^2)$.\\

Now we show that for any stopping time $\tau \in [s, h]$ we have in the sense of weak convergence in $\mathbb{L}^2(\mathbb R)$:
\begin{equation}\label{eq:convergence:phin:weak}
\int_s^{\tau }\Phi_n(Z^{a, n;s, x}_{2,u}) \lambda^{n}(Z^{b,n;s, x}_{2,u})\mathbf{1}_{Z^{b;s, x}_{2,u} =  0}\mathrm{d}u \underset{n \rightarrow  +\infty}{\rightarrow }\int_s^{\tau} Z^{a;s, x}_{2,u} \vartheta_u \mathbf{1}_{Z^{b;s, x}_{2,u}}  \mathrm{d}u.
\end{equation}
We have
\begin{eqnarray*}
\int_s^{\tau }\Phi_n(Z^{a, n;s, x}_{2,u}) \lambda^{n}(Z^{b,n;s, x}_{2,u})\mathbf{1}_{Z^{b;s, x}_{2,u} =  0}\mathrm{d}u &=& \int_s^{\tau }(\Phi_n(Z^{a, n;s, x}_{2,u}) - Z^{a;s, x}_{2,u}) \lambda^{n}(Z^{b,n;s, x}_{2,u}))\mathbf{1}_{Z^{b,s, x}_{2,u} =  0}\mathrm{d}u \\
&~&+ \int_s^{\tau } Z^{a;s, x}_{2,u} \lambda^{n}(Z^{b,n;s, x}_{2,u})\mathbf{1}_{Z^{b;s, x}_{2,u} =  0}\mathrm{d}s.
\end{eqnarray*}
The first term in the previous equality converges to $0$ in $\mathbb{L}^2(\mathbb R)$ by dominated convergence therefore it converges weakly. Now we show that the second term converges weakly. We prove that for any random variable $\xi \in \mathbb{L}^2(\mathbb R)$ and $\mathcal F_h$-measurable the following convergence holds
\begin{equation}\label{convergence:xi:weak}
\mathbb{E}[\xi \int_s^{\tau } Z^{a;s, x}_{2,u} \lambda^{n}(Z^{b,n;s, x}_{2,u})\mathbf{1}_{Z^{b;s, x}_{2,u} =  0}\mathrm{d}u] \underset{n \rightarrow  +\infty}{\rightarrow } \mathbb{E}[\xi \int_s^{\tau } Z^{a;s, x}_{2,u} \vartheta_u \mathbf{1}_{Z^{b;s, x}_{2,u} =0} \mathrm{d}u].
\end{equation}
Using a martingale decomposition result for martingales associated to jump processes, see \cite{davis1976representation}, to the conditional expectation of $\xi$ with respect to the filtration $\mathcal F$ we have
$$
\mathbb E[\xi|\mathcal F_\tau] = \mathbb{E}[\xi] + \int_s^{\tau} \Lambda_u \cdot \mathrm{d}M_u
$$
for some $\Lambda \in \mathcal{H}^2_{s,h}(\mathbb{R}^2)$. Consequently
\begin{eqnarray*}
\mathbb{E}[\xi \int_s^{\tau}Z^{a;s, x}_{2,u} \lambda^{n}(Z^{b,n;s, x}_{2,u})\mathbf{1}_{Z^{b;s, x}_{2,u} =  0}\mathrm{d}u]&=& \mathbb{E}[\int_s^{\tau} \Lambda_u \cdot \mathrm{d}M_u \int_s^{\tau}Z^{a;s, x}_{2,u} \lambda^{n}(Z^{b,n;s, x}_{2,u})\mathbf{1}_{Z^{b;s, x}_{2,u} =  0}\mathrm{d}u]\\
&~& + \mathbb{E}[\xi]\mathbb{E}[ \int_s^{\tau}Z^{a;s, x}_{2,u} \lambda^{n}(Z^{b,n;s, x}_{2,u})\mathbf{1}_{Z^{b;s, x}_{2,u} =  0}\mathrm{d}u].
\end{eqnarray*}
Notice moreover that
$$
\mathbb{E}[\xi]\mathbb{E}[ \int_s^{\tau}Z^{a;s, x}_{2,u} \lambda^{n}(Z^{b,n;s, x}_{2,u})\mathbf{1}_{Z^{b;s, x}_{2,u} =  0}\mathrm{d}u] \underset{n \rightarrow + \infty}{\rightarrow} \mathbb{E}[\xi]\mathbb{E}[ \int_s^{\tau}Z^{a;s, x}_{2,u} \vartheta_u \mathbf{1}_{Z^{b;s, x}_{2,u} =0}\mathrm{d}u]
$$
since $\lambda^{n}(Z^{b,n;s, x}_{2,u})\mathbf{1}_{Z^{b;s, x}_{2,u} =  0}$ converges to $\vartheta_u \mathbf{1}_{Z^{b;s, x}_{2,u} =0}$ and since $Z^{a;s, x} \in \mathcal{H}^2_{s,h}(\mathbb{R}^2)$. Using Ito's formula, we get 
\begin{align*}
&\mathbb{E}[\int_s^\tau \Lambda_u \cdot \mathrm{d}M_u \int_s^{\tau}Z^{a;s, x}_{2,u} \lambda^{n}(Z^{b,n;s, x}_{2,u})\mathbf{1}_{Z^{b;s, x}_{2,u} =  0}\mathrm{d}u] \\
&= \mathbb{E}\Big[\int_s^{\tau}\big( \int_s^u Z^{a;s, x}_{r, 2} \lambda^{n}(Z^{b,n;s, x}_{r, 2})\mathbf{1}_{Z^{b;s, x}_{r, 2} =  0}\mathrm{d}r\big) \Lambda_u\cdot \mathrm{d}M_u\Big]\\
& + \mathbb{E}\Big[\int_s^{\tau} \int_s^u  \Lambda_r\cdot \mathrm{d}M_r Z^{a;s, x}_{2,u} \lambda^{n}(Z^{b,n;s, x}_{2,u})\mathbf{1}_{Z^{b;s, x}_{2,u} =  0}\mathrm{d}u\Big].
\end{align*}
The first term is equal to zero. Concerning the second term, we set $\psi_r = \int_s^r \Lambda_u\cdot \mathrm{d}M_u$. Hence, for any $\kappa\geq 0$
\begin{align*}
&\mathbb{E}[\int_s^{\tau} \psi_u Z^{a;s, x}_{2,u} (\lambda^{n}(Z^{b,n;s, x}_{2,u})\mathbf{1}_{Z^{b;s, x}_{2,u} =  0} - \vartheta_u )\mathrm{d}u]\\
 &= \mathbb{E}[\int_s^{\tau} \psi_u Z^{a;s, x}_{2,u} \mathbf{1}_{|\psi_u Z^{a;s, x}_{2,u}|< \kappa }(\lambda^{n}(Z^{b,n;s, x}_{2,u}) - \vartheta_u )\mathbf{1}_{Z^{b;s, x}_{2,u} =  0}\mathrm{d}u] \\
 &+\mathbb{E}[\int_s^{\tau} \psi_u Z^{a;s, x}_{2,u}\mathbf{1}_{|\psi_u Z^{a;s, x}_{2,u}|\geq \kappa } (\lambda^{n}(Z^{b,n;s, x}_{2,u})\ - \vartheta_u) \mathbf{1}_{Z^{b;s, x}_{2,u} =  0}\mathrm{d}u] .
\end{align*}
The first term in the previous expression converges to $0$ since $\lambda^{n}(Z^{b,n;s, x}_{2,\cdot})\mathbf{1}_{Z^{b;s, x}_{2,\cdot} =  0}$ converges weakly towards $\vartheta $. The second one goes to zero when $\kappa$ goes to infinity as $\psi \|Z^{a;s, x}\|_2$ is in $\mathcal{H}^2_{s,h}(\mathbb{R})$. We have proved the convergence \eqref{convergence:xi:weak}. Hence, the convergence \eqref{eq:convergence:phin:weak} holds weakly in $\mathbb{L}^2(\mathbb R)$.\\

We deduce that $
\int_s^{\tau}H^{a, n}(u, X^{s,x}_u) \mathrm{d}u$ converges weakly to $\int_0^{\tau}H^{a, \star}(Z^{a;s, x}_u, Z^{b;s, x}_u, \vartheta_u) \mathrm{d}u$ in  $\mathbb{L}^2(\mathbb R)$ along the subsequence $(n_k)_{k\geq 0}$.\\

\textbf{Step 2d. Convergence to the solution of a bang-bang BSDE}\\

\noindent If we write the first BSDE in the system $\mathbf{(J^n)}$ in a forward way, we get
\begin{equation*}
Y^{a, n;s,x}_{\tau} = Y^{a, n;s,x}_s - \int_s^{\tau}H^{a, n}(u, X^{s,x}_{u^-}) \mathrm{d}u + \int_s^{\tau}Z^{a, n;s, x}_u\mathrm{d}M_u.
\end{equation*}
We recall that we write $n$ instead of $n_k$ so that all the convergence that we obtain has to be understood up to a subsequence. Thus, from the almost sure and $\mathcal{S}^2_{s,h}(\mathbb{R})$ convergence of $(Y^{a, n;s,x})_{n \geq 0}$ to $Y^{a;s,x}$ together with
$$
\int_s^{\tau} Z^{a, n;s, x}_u\cdot\mathrm{d}M_u \underset{n \rightarrow + \infty}{\rightarrow}\int_s^{\tau} Z^{a;s, x}_u\cdot\mathrm{d}M_u
,~\text{in }\mathbb{L}^2(\mathbb R),
$$
and the convergence of the generator $H^{a,n}$ proved in Step 2c, we deduce that
$$
Y^{a;s,x}_{\tau} = Y^{a;s,x}_s - \int_s^{\tau}H^{a, \star}(Z^{a;s, x}_u, Z^{b;s, x}_u, \vartheta_u) \mathrm{d}u + \int_s^{\tau}Z^{a;s, x}_u\mathrm{d}M_u,\; \mathbb{P}-a.s.
$$
This result being true for any stopping time $\tau \in [s, h],$ the processes on both sides are indistinguishable and we have
$$
\mathbb{P}-a.s.~Y^{a;s,x}_{u} = Y^{a;s,x}_s - \int_s^{u}H^{a, \star}(Z^{a;s, x}_r, Z^{b;s, x}_r, \vartheta_r) \mathrm{d}r + \int_s^{u}Z^{a;s, x}_r\mathrm{d}M_r,\; \forall u \in [s, h].
$$
Finally we have
$$
\mathbb{P}-a.s.~Y^{a;s,x}_{u} = g^a(X^{s, x}_h) + \int_u^{h}H^{a, \star}(Z^{a;s, x}_r, Z^{b;s, x}_r, \vartheta_r) \mathrm{d}r - \int_u^{h}Z^{a;s, x}_r\mathrm{d}M_r, ~ \forall u \in [s, h].
$$
with $Y^{a;s, x}\in \mathcal{S}^{2}_{s,h}(\mathbb{R})$ and $Z^{a;s, x} \in \mathcal{H}^2_{s,h}(\mathbb{R}^2) $. We have the same result by considering the index $b$ and by denoting $\theta_u$ the almost sure limit of $(\lambda^n(Z^{a, n;s, x}_{u})\mathbf{1}_{ Z^{a, n;s, x}_{u} = 0} )_{n\geq 0}$ which holds also in $\mathcal{H}^2_{s,h} $ by the dominated convergence theorem. \\

\textbf{Step 3: Nash equilibrium and conclusion.}\\

We have seen in the previous step that we can build $\vartheta$ and $\theta$, which are functions of $(u,N_u^a,N_u^b)$ ensuring the existence of a solution a solution $(Y^a,Y^b,Z^a,Z^b)\in (\mathcal{S}^{2}_{s,h}(\mathbb{R}))^2\times( \mathcal{H}^2_{s,h}(\mathbb{R}^2))^2$ to the following coupled BSDE (by taking $s=0$), \begin{equation}
\label{eq:bsde}
\left\{\begin{array}{ll}
- \mathrm{d}Y^a_u =& H^{a,\star}( Z^a_u, Z^b_u, \vartheta_u) - Z^a_u\cdot \mathrm{d}M_u,~Y^a_h = g^a(X^{0, 0}_t)\\
- \mathrm{d}Y^b_u =& H^{b,\star}( Z^a_u, Z^b_u, \theta_u) - Z^b_u\cdot \mathrm{d}M_u,~Y^b_h = g^b(X^{0, 0}_t).
\end{array}\right.
\end{equation}
We could rely this BSDE to the system $\mathbf{(S)}$ and use Proposition \ref{th:nash_equilibrium_bangbang}. However, we are not able to prove the continuous differentiability of the functions $V^i$ with respect to the time variable. It is why we use the theory of BSDEs similarly to \cite{hamadene2014bang} with the proposition below to conclude.
	
	\begin{proposition}[Extension of Theorems 2.5 and 2.6 in \cite{hamadene2014bang}]
\label{th:nash_equilibrium}
There exist a pair of deterministic functions $V^a, V^b$ and some adapted processes $\vartheta$ and $\theta$ with values in $[\lambda_-,\lambda_+]$ such that 
\begin{enumerate}
\item[$\bullet$] BSDE \eqref{eq:bsde} admits a solution $(Y^a, Y^b,Z^a,Z^b)\in(\mathcal{S}^{2}_{h}(\mathbb{R}))^2\times (\mathcal{H}^2_{h}(\mathbb{R}^2))^2$, 

\item[$\bullet$]  $V^a$ and $V^b$ are two deterministic measurable functions with polynomial growth from $[0, h]\times \mathbb{R}^2$ to $\mathbb{R}$ such that $\mathbb{P}-$as, $\forall u \leq h,~Y^a_u =V^a(u, X_u)$ and $Y^b_u =V^b(u, X_u)$.
\item[$\bullet$]  The pair of controls $(\lambda^{\star}_a(Z^a_u, \theta_u),\lambda^{ \star}_b(Z^b_u, \vartheta_u))_{u\leq t} $ defined by $(\mathbf{L})$ where $\vartheta$ and $\theta$ are obtained as an almost sure (up to a subsequence) and $\mathcal H^2_h(\mathbb R^2)$ limits of $\lambda^n(Z^{b,n}_u)\mathbf 1_{Z_u^{b}=0}$ and $\lambda^n(Z^{a,n}_u)\mathbf 1_{Z_u^{a}=0}$ respectively is a bang-bang type Nash equilibrium point of the non zero-sum stochastic differential game \eqref{eq:market_taker_game}. 
\end{enumerate}
\end{proposition}
\begin{proof}
Properties $1.$ and $2.$ are direct consequences of the proof made in Step 2. Property 3. is obtained by adapting the proof of Proposition 2.4 in \cite{hamadene2014bang} to the jump case, with minimizations instead of maximizations.
\end{proof}
Hence, Step 1 provides that the system $\mathbf{(S^n)}$ admits a unique viscosity solution given by the unique solution of $\mathbf{(\widetilde{J^{n}})}$ which approaches the solution of \eqref{eq:bsde} so that $\lambda^n(Z^{b,n}_u)\mathbf 1_{Z_u^{b}=0}$ and $\lambda^n(Z^{a,n}_u)\mathbf 1_{Z_u^{a}=0}$ converge almost surely up to a subsequence (and in fact in $\mathcal H^2_h(\mathbb R^2)$) to a Nash equilibrium for the game \eqref{eq:market_taker_game} by using Proposition \ref{th:nash_equilibrium}. This concludes the proof of Theorem \ref{thm:edpapproach}.

\subsection{Proof of Corollary \ref{cor:Eh} and numerical method}
\label{appendix:numerical_nash}
In Theorem \ref{thm:edpapproach} we only get convergence results up to a subsequence. However numerically we observe that the sequence $(V^{i, n})_{n\geq 0}$ converges for $i=a$ or $b$. Therefore to approach the solution of the system $\mathbf{(S)}$ we solve the approached system $\mathbf{(S^n)}$ for $n$ large. To implement the numerical method we need to bound the domain. In practice this means that there is only a limited number of orders in auctions. Thus we consider the new system 
$$
\mathbf{(S_Q^{n})}\left\{
\begin{array}{ll}
\partial_s V^{a,n}+H^{a,n}(D^Q_aV^{a,n},D^Q_bV^{a,n}, D^Q_bV^{b,n})=0,\; s\in [0,h),\, (\alpha,\beta)\in \{0, \dots, Q\}^2,&\\
V^{a,n}(h,\alpha, \beta)=g^a(\alpha, \beta),&\\
\partial_s V^{b,n}+H^{b,n}(D^Q_bV^{b,n},D^Q_aV^{b,n}, D^Q_aV^{a,n}))=0,\,  s\in [0,h),\, (\alpha,\beta)\in \{0, \dots, Q\}^2,&\\
V^{b,n}(h,\alpha,\beta)=g^b(\alpha, \beta),&
\end{array}
\right.
$$
on the domain $[0, h] \times \{0, \dots, Q\}^2$. The operators $(D_a^Q,D_b^Q)$ are defined similarly to $(D_a,D_b)$ with the following boundary conditions
$$
D^Q_aV(s, Q, m) = 0 \text{ and }D^Q_bV(s, n, Q) = 0\text{ for any }(s,n,m)\in [0,h]\times  \{0, \dots, Q\}^2.
$$
Interpreting $\mathbf{(S_b^n)}$ as an ordinary differential equation in $\mathbb{R}^{(Q+1)^2}$ according to Cauchy-Lipschitz Theorem we have existence of a solution $(V^{a, n}_Q, V^{b, n}_Q)$ for the system $\mathbf{(S_Q^n)}$ which is unique.\\

Remember that in our model the auction starts at time $\tau = \inf\{s>0\text{ s.t. } N^a_s + N^b_s>0 \}$. Consequently market takers optimize their behavior by controlling the processes $(N^a_{\tau + \cdot}, N^b_{\tau + \cdot} )$. Now remark that
	$$
	I_{\tau + h}^2 = N^a_{\tau + h}(N^a_{\tau + h} - N^b_{\tau + h}) + N^b_{\tau + h}(N^b_{\tau + h} - N^a_{\tau + h}).
	$$
	Consequently, the symmetry of the problem with respect to $a$ and $b$ leads to
	$$
	\mathbb{E}[I_{\tau + h}^2] =  \mathbb{P}(N^a_{\tau} = 1)\big( V^a(0, 1, 0) +V^b(0, 1, 0) \big) + \mathbb{P}(N^b_{\tau} = 1)\big( V^a(0, 0,1)  +V^b(0, 0, 1) \big).
	$$
	Now we assume that market takers controls their intensities using a pair of Nash Equilibrium controls $(\lambda^{\star}_a, \lambda^{\star}_b)$ obtained in Theorem \ref{thm:edpapproach} as limit of the smoothed problem.  According to the first point of Theorem \ref{thm:edpapproach} and since $V^a(0,0,1)=V^b(0,1,0)$ and $V^b(0,0,1)=V^a(0,1,0)$, we get Corollary \ref{cor:Eh} so that
	$$
	\mathbb{E}[I_h^2] = \underset{n\rightarrow + \infty}{\lim } V^{a, n}(0, 1, 0)+V^{b, n}(0, 1, 0)=V^{a}(0, 1, 0)+V^{b}(0, 1, 0).
	$$
	Let $\bar{V}^{a, n}$ (resp. $\bar{V}^{b, n}$) be defined as the backward form of the solutions $V^{a, n}$ (resp. $V^{b, n}$) of $\mathbf{(S^n)}$, more precisely
	$$
	\bar{V}^{i, n}(s, \cdot, \cdot) = V^{i, n}(h-s, \cdot, \cdot),\; s\in [0,h], \text{ for }i \in \{a, b\}.
	$$
	In the same way, we denote by $(\bar{V}_Q^{a, n}, \bar{V}_Q^{b, n})$ the backward versions of the solution $(V_Q^{a, n}, V_Q^{b, n})$ of $\mathbf{(S_Q^n)}$. The functions $(\bar{V}_Q^{a, n}, \bar{V}_Q^{b, n})$ are computed by solving the backward system $\mathbf{(S_Q^n)}$.\\
	
		Finally note that 
		$$
		\mathbb{E}[I_h^2] = \underset{n\rightarrow + \infty}{\lim }  \bar{V}^{a, n}(h, 1, 0)+ \bar{V}^{b, n}(h, 1, 0)\approx \bar{V}_Q^{a, n}(h, 1, 0)+ \bar{V}_Q^{b, n}(h, 1, 0).
		$$
		Hence we use the quantity $\bar{V}_Q^{a, n}(h, 1, 0)+\bar{V}_Q^{b, n}(h, 1, 0)$ for $n=1 000$ and $Q=100$ to approach more accurately $\mathbb{E}[I_{h}^2]$.

\section{Model extension: Market makers can cancel their limit orders }
\label{appendix:cancellation}
We can extend our model and allow market makers to revise their position before the auction clearing by cancelling their limit orders. Formally a market maker arrived at time $\tau \leq \tau^{cl}_i$ will maintain its position until the auction clearing at time $t$ with a probability $\theta(t-\tau^{cl}_i)$, where $\theta$ is a $[0,1]$-valued decreasing function such that $\theta(0)=1$. Hence, the number of market makers present at the $i-th$ auction clearing is 
$$
\tilde{N}_{\tau^{cl}_{i - 1}} - \tilde{N}_{\tau^{cl}_i} ~, ~~\text{with} ~~\tilde{N}_s = \sum_{j = N^{mm}_{\tau^{cl}_{i - 1}} + 1 }^{N^{mm}_{\tau^{cl}_i}} \mathbf{1}_{X_k \leq \theta(\tau_k-\tau^{cl}_i)},
$$
where $(X_j)_{j\geq 0}$ is a sequence of i.i.d. random variables with uniform law on $[0, 1]$. We can show that during auction time $(\overline{N}_s)_{0\leq s \leq h} = (\tilde{N}_{\tau^{op}_i + s})_{s\geq 0}$ has the same law than an inhomogeneous Poisson process with intensity 
$$
\lambda(s) = \mu \theta(t - s).
$$
Moreover we still have an explicit formula for $E$. 
$$
E^{mid}(h) = (1 - e^{-m_h}\frac{\nu}{\nu + \mu})^{-1} e^{\nu h }\int_{h}^{+\infty} \nu e^{-\nu t } \Big( (\sigma^2_{f}\frac{t}{6}+\sigma^2)e^{-m_t}\int_{0}^{m_t}\frac{e^s - 1}{s}\mathrm{d}s +\sigma_f^2 \frac{t}{3}(1 - e^{-m_t})  \Big)\mathrm{d}t
$$
and
$$
E(h) = E^{mid}(h) + \frac{\mathbb{E}[I^2_{\tau^{op}_1 + h}]}{K^2} (1 - e^{-m_h}\frac{\nu}{\nu + \mu})^{-1} e^{\nu h } \int_h^{+\infty}\nu e^{-m_t }e^{-m_t}\int_0^{m_t} \frac{1}{s}\int_0^s \frac{e^u - 1}{u} \mathrm{d}u \mathrm{d}s \mathrm{d}t
$$
with 
$$
m_t = \int_0^{t}\mu  \theta(s) \mathrm{d}s.
$$

\section{Proof of Lemma \ref{lemma:utility}}
\label{appendix:regenerative}

Consider for any $s>\tau^{cl}_1$, $X_s = (\overline{P}^{cl}_s-\overline{P}_s)^2$. We show that $(X_s)_{s > \tau^{cl}_1}$ is a regenerative process with renewal times given by $(\tau^{cl}_i)_{i\geq 1}$.\\

Consider  $\tau^{cl}_{i} \leq s < \tau^{cl}_{i+1}$ we have
\begin{equation}
\overline{P}^{cl}_s-\overline{P}_s = \frac{1}{N^{i,mm}_{\Delta_i} }\sum_{k = 1}^{N^{i,mm}_{\Delta_i} } (P_{\tau_i^{cl}} - P_{\tau^{cl}_{i-1}+\tau^{i, mm}_k}) + \frac{1}{N^{i,mm}_{\Delta_i} } \sum_{k = 1}^{ N^{i,mm}_{\Delta_i} } g_k + \frac{I^i_{ \Delta_i} }{K N^{i,mm}_{\Delta_i}}.
\label{eq:proof_regen_a}
\end{equation}
According to Assumption \ref{assumption:order_flow} the process $ (N^{i,mm}_{t}, I^i_t)_{t\geq 0}$
is independent from $\mathcal{F}_{\tau_{i-1}^{cl}}$ with same law as $(N^{mm}_{t}, I_t)_{t\geq 0}$. Same results holds for $(P_{\tau_{i-1}^{cl}+t} - P_{\tau_{i-1}^{cl}})_{t\geq 0}$ and $(P_{t}-P_0)_{t\geq 0}$ since $P$ is a Brownian motion. Consequently $N^{i, mm}_{\Delta_{i}}$, $I^i_{\Delta_{i}}$ and $(P_{\tau_{i-1}^{cl}+t} - P_{\tau_{i-1}^{cl}})_{t\geq 0}$ are independent from $(X_s)_{s<\tau_i^{cl}}$ with same law as $N^{mm}_{\tau_{1}^{cl}}$, $I_{\tau^{cl}_1}$ and $(P_{t}-P_0)_{t \geq 0}$.\\

Thus according to \eqref{eq:proof_regen_a} and since $X$ is piecewise continuous with jump at times $(\tau^{cl}_i)_{i\geq  1 }$, for any $ \tau_{i}^{cl} \leq  s < \tau_{i+1}^{cl} $, $X_s$ is independent of $(X_s)_{ s < \tau^{cl}_i}$ and has the same distribution than $X_{\tau^{cl}_1}$. Thus $X$ is regenerative with renewal times equal to $(\tau^{cl}_i)_{i \geq >1}$\\

Thus according to Theorem 3.1 Chap VI in \cite{asmussen2008applied} we have the almost sure convergence
\begin{eqnarray*}
\frac{\int_{0}^{t}X_s\mathrm{d}s}{t}\underset{t\rightarrow +\infty}{\rightarrow}& &\frac{\mathbb{E}[\int_{\tau^{cl}_{1}}^{\tau^{cl}_{2}}X_s\mathrm{d}s]}{\mathbb{E}[\tau^{cl}_{2} - \tau^{cl}_{1}]}\\
&=& \frac{\mathbb{E}[\tau^{cl}_{2} - \tau^{cl}_{1}] \mathbb{E}[X_{\tau_1^{cl}}]}{\mathbb{E}[\tau^{cl}_{2} - \tau^{cl}_{1}]}\\
&=& \mathbb{E}[X_{\tau_1^{cl}}] = \mathbb{E}[(P_{\tau_1^{cl}} - P^{cl}_{\tau_1^{cl}})^2  ].
\end{eqnarray*}
Thus we get the stated result.

\section*{Acknowledgments}
The authors gratefully acknowledge the financial supports of the ERC Grant 679836 Staqamof, the Chaires Analytics and Models for Regulation and Financial Risk. Thibaut Mastrolia acknowledges the financial support of the ANR project PACMAN.

\bibliographystyle{apalike}

\end{document}